\documentclass{preprint}  %

 \usepackage[colorlinks=true]{hyperref}  %
\usepackage{mathtools}  %
\usepackage{graphicx}  %
\usepackage{csquotes}  %
\usepackage{amsmath}  %
\usepackage{amssymb}  %
\usepackage{booktabs}  %
\usepackage{hyphenat}  %
\usepackage{caption}  %
 \usepackage{subcaption}  %
\usepackage{bm}

 \usepackage[capitalize]{cleveref} %
\crefname{section}{\S}{\S\S}
\crefname{appendix}{App.}{Apps.}
\crefname{theorem}{Thm.}{Thms.}
\crefname{lemma}{Lem.}{Lems.}
\crefname{claim}{Clm.}{Clms.}
\crefname{definition}{Def.}{Defs.}
\crefname{remark}{Rem.}{Rems.}

\usepackage{xparse}
\usepackage{xcolor}
\providecolor{mypink}{rgb}{0.785,0.258,0.465}
\providecolor{myorange}{rgb}{0.855,0.426,0.023}

\DeclarePairedDelimiter\set{\{}{\}}

\mathchardef\mhyphen="2D %

\newcommand{\eqq}{\coloneqq}

\newcommand\RR{\mathbb{R}}

\newcommand\Ac{\mathcal{A}}
\newcommand\Bc{\mathcal{B}}
\newcommand\Cc{\mathcal{C}}

\newcommand\Hc{\mathcal{H}}

\newcommand\Rc{\mathcal{R}}

\newcommand{\advy}{\Ac}  %

\NewDocumentCommand{\Prob}{o m}{%
        \IfValueTF{#1}{%
                \operatornamewithlimits{Pr}_{#1}%
        }{%
                \operatorname{Pr}%
        }%
        \mathopen{}\left[#2\right]\mathclose{}
}

\newcommand\blockk{\text{\scalebox{0.5}{$\blacksquare$}}}
\newcommand{\weight}{\Gamma}  %
\newcommand{\cweight}{\overline{\Gamma}}  %
\newcommand{\bcweight}{\overline{\Gamma}_\blockk}  %
\newcommand\bhsp{\bm S^{\Hc}} %
\newcommand\basp{\bm S^{\Ac}} %
\newcommand\asptil{{\widetilde S}^{\Ac}} %
\newcommand\basptil{\bm {\widetilde S}^{\Ac}} %
\newcommand\bhvdf{\bm V^{\Hc}}  %
\newcommand\bavdf{\bm V^{\Ac}}  %
\newcommand\avdftil{{\widetilde V}^{\Ac}}  %
\newcommand\bavdftil{\bm {\widetilde V}^{\Ac}}  %
\newcommand\bhw{\bm W^{\Hc}} %
\newcommand\baw{\bm W^{\Ac}} %
\newcommand\awtil{{\widetilde W}^{\Ac}} %
\newcommand\bawtil{\bm {\widetilde W}^{\Ac}} %
\newcommand\RRgez{\RR_{\geq 0}} %
\newcommand\RRgz{\RR_{> 0}} %

\newcommand{\chain}{\Cc\Cc}  %
\newcommand{\hchain}{\chain^{\Hc}} %
\newcommand{\achain}{\chain^{\Ac}} %
\newcommand{\bchain}{\Bc\Cc}  %
\newcommand{\hbchain}{\bchain^{\Hc}} %
\newcommand{\abchain}{\bchain^{\Ac}} %

\newcommand{\ttil}{{\widetilde t}}
\newcommand{\Ttil}{{\widetilde T}}
\newcommand{\Ttilend}{{\Ttil_\eend}}
\newcommand{\til}{\widetilde}
\newcommand{\Vb}{\bm V}
\newcommand\Sb{\bm S}
\newcommand\Wb{\bm W}
\newcommand{\vb}{\bm v}
\newcommand{\ssb}{\bm s}
\newcommand{\wb}{\bm w}
\newcommand{\bVb}{\textbf{V}_\blockk}
\newcommand\bSb{\textbf{S}_\blockk}
\newcommand\bWb{\textbf{W}_\blockk}
\newcommand{\bVbmin}{\textbf{V}_{\min}}
\newcommand\bSbmin{\textbf{S}_{\min}}
\newcommand\bWbmin{\textbf{W}_{\min}}
\newcommand{\bVbmax}{\textbf{V}_{\max}}
\newcommand\bSbmax{\textbf{S}_{\max}}
\newcommand\bWbmax{\textbf{W}_{\max}}

\newcommand\res{\Rc}
\newcommand\ares{\Rc^\Ac}
\newcommand\hres{\Rc^\Hc}
\newcommand{\eend}{\mathsf{end}} 
\newcommand{\at}{\mathsf{AT}} %
\newcommand{\atinv}{\mathsf{AT}^{-1}} %

\newcommand{\stake}{\text{\textit{St}}}

\title{Nakamoto Consensus from Multiple Resources}

 \author[1]{Mirza Ahad Baig\orcidlink{0000-0003-3650-7893}}
 \author[1]{Christoph U. Günther\orcidlink{0009-0001-5790-695X}}
 \author[1]{Krzysztof Pietrzak}

 \affil[1]{Institute of Science and Technology Austria\\ \email{\{mirzaahad.baig, cguenthe, pietrzak\}@ista.ac.at}}

\begin{document}
\maketitle

\begin{abstract}
        The blocks in the Bitcoin blockchain \enquote{record} the amount of work $W$ that went into creating them through proofs of work. 
        When honest parties control a majority of the work, consensus is achieved by picking the chain with the highest recorded weight.
        Resources other than work have been considered to secure such longest-chain blockchains.
        In Chia, blocks record the amount of disk-space  $S$ (via a proof of space) and \emph{sequential} computational steps $V$ (through a VDF).

        In this paper, we ask what weight functions $\Gamma(S,V,W)$ (that assign a weight to a block as a function of the recorded space, speed, and work) are secure 
        in the sense that whenever the weight of the resources controlled by honest parties is larger than the weight of adversarial parties, the blockchain is secure against private double-spending attacks.

        We completely classify such functions in an idealized ``continuous'' model: $\Gamma(S,V,W)$ is secure against private double-spending attacks if and only if it is homogeneous of degree one in the \enquote{timed} resources $V$ and $W$, i.e., 
        $\alpha\Gamma(S,V,W)=\Gamma(S,\alpha  V, \alpha W)$.
        This includes the Bitcoin rule $\Gamma(S,V,W)=W$ and the Chia rule $\Gamma(S,V,W) = S \cdot V$. In a more realistic model where blocks are created at discrete time-points, one additionally needs some mild assumptions on the dependency on $S$ (basically, the weight should not grow too much if $S$ is slightly increased, say linear as in Chia).

        Our classification is more general and allows various instantiations of the same resource.
        It provides a powerful tool for designing new longest-chain blockchains. 
        E.g., consider combining different PoWs to counter centralization, 
        say the Bitcoin PoW $W_1$ and a memory-hard PoW $W_2$. 
        Previous work suggested to use $W_1+W_2$ as weight.
        Our results show that using e.g., $\sqrt{W_1}\cdot \sqrt{W_2}$  or $\min\{W_1,W_2\}$ are 
        also secure, and we argue that in practice these are much better choices.
\end{abstract}
\section{Introduction}
Achieving consensus in a permissionless setting is a famously difficult problem.
Nakamoto solved it by introducing the Bitcoin blockchain~\cite{bitcoin} that achieves consensus on a chain of blocks by having parties expend a \emph{resource}: parallelizable computation (commonly called \emph{work}).

In Bitcoin, appending a block to a chain requires a \emph{proof-of-work} (PoW), i.e., solving a computationally-expensive puzzle.
This puzzle is designed such that each block represents the (expected) amount of computation that was expended to append it.
As a consequence, each chain represents the total amount of computation required to create it.
This allows for a simple consensus mechanism commonly called the \emph{longest-chain rule}:
Given two different chains, pick the one that required more computation to create.
Note that a more accurate term is \emph{heaviest-chain rule}, which we will use interchangeably throughout the paper.

While this design achieves consensus, more importantly it also achieves a property called \emph{persistence}~\cite{EC:GarKiaLeo15} under a simple economic assumption:
As long as honest parties control more than half of the computational resources committed to Bitcoin, a block that has been part of the chain for some time will always be part of the chain.
Since Bitcoin blocks contain transactions, this effectively means that an adversary cannot \emph{double-spend} a coin.

While Bitcoin's design is simple, its reliance on PoW has its flaws.
For example, it wastes a lot of energy, and the manufacturing of the PoW hardware has become increasingly centralized.
Amongst other reasons, this has lead to the development of other blockchain protocols.
These protocols can be broadly categorized along three axes.

\begin{description}
        \item[The Underlying Resources]
                Bitcoin relies on parallelizable computation, which is a \emph{physical resource}.
                Two natural alternatives are \emph{disk space} and \emph{sequential computation}.
                A different class of resources is not physical, but \emph{on-chain}~\cite{ACLVZ23,permissionlessconsensus}.
                The most well-known is \emph{stake}, which comprises different approaches that essentially rely on the on-chain coin balance of a party.
        \item[The Consensus Design]
                A broad distinction is between Byzantine\hyp{}fault\hyp{}tolerant\hyp{}style (e.g., Algorand~\cite{algorand} or Filecoin~\cite{filecoin}) and longest-chain protocols.
                Within themselves, longest-chain protocols differ in their \emph{fork-choice rule}, which prescribes how they select the longest chain. 
                Some are \emph{Nakamoto-like} and, like Bitcoin, pick the heaviest chain, i.e., the one whose blocks cumulatively required the most resources to create (e.g., Chia~\cite{chia}).
                Others rely on more complex fork-choice rules that, e.g., take into account where two chains fork (e.g., Ouroboros~\cite{CCS:BGKRZ18}). %

        \item[The Degree of Permissionlessness]
                Roughgarden and Lewis-Pye~\cite{permissionlessconsensus} observe that \enquote{permissionless} is colloquially used to describe different settings that vary in how permissionless they are.
                \emph{Fully-permissionless} protocols function obliviously to current protocol participants (e.g., Bitcoin or Chia).
                These differ from protocols that require some information about participants (e.g., how many coins they are staking in Algorand, or commitments to disk space in Filecoin).
                While the latter are still permissionless in the sense that anyone can participate, they impose stricter requirements on participants.
\end{description}

\subsection{Our Contributions}
In this paper, we completely characterize the design space of \emph{Nakamoto-like} protocols operating in the \emph{fully-permissionless} setting using the \emph{physical resources} disk space, sequential computation, and parallelizable computation that are secure against \emph{private double-spending} attacks.

We observe that Nakamoto-like protocols only differ in what resources their blocks record, and---especially if multiple resources are used---how they decide which of two blocks required \emph{more resources} to create.
We model these differences using an abstraction called the \emph{weight function} $\Gamma\colon \RR^3 \to \RR$. 
It takes as input the three resources possibly recorded by a block, i.e., disk space $S$, sequential computation $V$\footnote{Think $V$ as in velocity of the sequential computation or $V$ as in verifiable delay function (VDF), the cryptographic primitive usually used to capture the number of sequential computation steps.} and parallelizable computation/work $W$, and outputs the \emph{weight} $\Gamma(S, V, W)$ of a block.
In the context of weight functions, the heaviest-chain rule now picks the chain with the highest weight where a chain's weight is defined as the sum of the weight of all its blocks.

To get an intuition for the weight function abstraction, let us provide some examples.
The weight function $\Gamma_\text{Bitcoin}(S, V, W) = W$ describes Bitcoin (or any other similar PoW-based Nakamoto-like chain, e.g., Litecoin\footnote{\url{https://litecoin.org/}}).
More interesting is Chia~\cite{chia}, a Nakamoto-like chain combining disk space and sequential computation following $\Gamma_\text{Chia}(S, V, W) = S \cdot V$.

Our main result, informally stated in Theorem~\ref{thm:main} below, fully characterizes which weight functions result in a Nakamoto-like blockchain that is secure against private double-spending attacks~\cite{bitcoin,CCS:DKTTVW20} in the fully-permissionless setting~\cite{permissionlessconsensus}.

In this work we address double spending, but not economic attacks such as selfish mining~\cite{FC:EyaSir14}.
Such attacks are an orthogonal issue and require a different set of tools. Preventing double spending gives some additional guarantees, like the fact that one can trust the timestamps on the chain~\cite{EPRINT:TziSriZin23}. 

To achieve a broad and simple characterization, we necessarily have to abstract implementation details and generalize over different blockchain designs. We operate under the maxim that a good design principle is that chains should reflect the resources that went into creating them. In practical instantiation of a blockchain this may not be fully guaranteed. Network delay and limited block space lead to resources only being approximated by blocks.

Network delays have been well-studied for single resource blockchains like Bitcoin, Chia and Ethereum~\cite{CCS:DKTTVW20,C:GazRenRus23,AFT:GuoRen22,CCS:GazRenRus22} and they add a multiplicative factor $\chi(\Delta) < 1$ ($\Delta$ is the network delay) to the honest resource in the honest majority assumption, where $\chi$ depends on the particular blockchain. Thus taking network delays into account would only quantitatively affect our honest majority assumption and not give any new interesting insights. As for resources being accurately reflected on chain: this also depends on the precise implementation of the chain. There are multiple options: either take an average over multiple blocks (akin to how Bitcoin difficulty changes), or put multiple proofs into one block/epoch (for example, including the top $k$ partial PoW solutions in a block, or like in Chia where multiple PoSpace blocks come from the same challenge). This leads to a good approximation of the total resources available at any point of time. To compensate for any loss we again need to include a multiplicative factor to the honest resources in the honest majority assumption. The precise formula would depend on the exact implementation details. Since our focus is on a unified idealized model, we abstract away these details and leave the question of best practical design and trade-offs involved in it as future work. There are other attacks like grinding and double dipping\footnote{A high level overview on these attacks can be found on \url{https://docs.chia.net/longest-chain-protocols/}} against chains that use space, but we have techniques to prevent them~\cite{FC:PKFGAP18,BDKOTVWZ22}. Thus we assume they are implicitly taken care of in the design.
One issue are so called replotting attacks. As we'll discuss in \S\ref{S:replot}, in practice replotting can be prevented with a careful design putting bounds on the total weight of individual blocks. Since replotting is not as well understood as the other issues, we will explicitly exclude replotting in the statement of the theorem below. 
We'll discuss our model in more detail in \S\ref{S:model}. 
\begin{theorem}[Main, Informal]\label{thm:main}
        In the fully-permissionless~\cite{permissionlessconsensus} setting and ignoring replotting attacks, a Nakamoto-like blockchain is secure against private double-spending attacks under the honest majority assumption (cf. below) \emph{if and only if} the weight function $\Gamma(S, V, W)$ fulfills the following conditions:
        \begin{enumerate}
                \item \textbf{Monotone:} $\Gamma$ is monotonically increasing\label{enum:thm:main:cond1}
                \item \textbf{Homogeneous in $V, W$:} $\alpha\Gamma(S, V, W) = \Gamma(S, \alpha V, \alpha W)$ for $\alpha > 0$\label{enum:thm:main:cond2}
        \end{enumerate}
The honest majority assumption states that at any point in time during the attack $\Gamma$ applied to the resources of the honest parties is larger than $\Gamma$ applied to the resources of the adversary.
\end{theorem}

\cref{thm:main} is in an ideal model. In~\cref{sec:discrete} we provide a less idealized, discrete model and prove a related result in this model which essentially states that every weight function that is insecure in ideal model is also insecure in the more realistic discrete model. On the other hand every weight function secure in the ideal model is secure in the discrete model with a slightly stronger honest majority condition. 

Finally, we deal with the replotting attacks and how to mitigate them in~\cref{S:replot}.

\subsection{Implications of our Result}\label{sec:intro:impl}
\subsubsection{Space-based Blockchains}
Three very different blockchain designs whose main resource is disk space are Chia~\cite{chia}, Filecoin~\cite{filecoin}, and Spacemint~\cite{FC:PKFGAP18}.
The first two are deployed and running in practice, while the latter is an academic proposal.

Of the above, Chia is the only one captured by our result, i.e., it is a Nakamoto-like protocol in the fully-permissionless setting.
Its weight function is $\Gamma_\text{Chia}(S, V, W) = S \cdot V$ which is secure against double-spending attacks by our \cref{thm:main}.

Spacemint was an early proposal of a fully-permissionless blockchain based solely on proofs of space, and thus cannot be secure against double-spending attacks according to \cref{thm:main}.  
The design of Spacemint slightly defers from Nakamoto's chain-selection rule as older blocks are given less weight than more recent ones, but even with this twist the security of Spacemint against double-spending only holds if the honest space never decreases by too much, and never increases too fast.

In contrast, Filecoin is not captured by our result because it is not fully-permissionless, but instead operates in the stronger quasi-permissionless setting (cf.~\cite{permissionlessconsensus}).\footnote{
        Filecoin is also not Nakamoto-like since it is a DAG-based protocol (using GHOST, the \emph{Greediest Heaviest-Observed Sub-Tree} rule~\cite{EPRINT:SomZoh13}) together with a finality gadget~\cite{filecoinf3}.
        Note that the finality gadget is not essential, and GHOST is the DAG-analogue to Bitcoin's longest/heaviest-chain rule.
        So, for our purposes, Filecoin could easily be modified to be Nakamoto-like (this has also been mentioned in~\cite{ACLVZ23}). 
        As we'll elaborate in \S\ref{S:replot}, it seems running a space based chain in the quasi-permissionless setting is quite expensive as to prevent reploting parties must constantly prove they hold the committed space and this proofs need to be recorded on chain.
} 
Since Filecoin's weight function is $\Gamma_\text{Filecoin}(S, V, W) = S$, \cref{thm:main} essentially shows that a setting stronger than the fully-permissionless one is necessary. 

Let us stress that, even in the fully-permissionless setting and when only relying on space, we only rule out secure constructions of Nakamoto-like blockchains (where the weight of a chain is the sum of the weights of its blocks, and the chain selection rule picks the heaviest chain). While most fully-permissionless blockchains are of this form, this does not rule out the possibility that a completely different chain selection rule would be secure. A recent work~\cite{eprint:BaigPie25} shows that, unfortunately, this is not the case, 
and no such chain-selection rule exists. It gives a concrete attack against any chain-selection rule, and an almost matching lower-bound, i.e., a concrete (albeit very strange) chain-selection rule for which this attack is basically optimal.

\subsubsection{Combining multiple PoWs}
Since the production of Bitcoin mining-hardware has become increasingly centralized, one might consider combining two different PoWs, e.g., $W_1$ using SHA256 and $W_2$ using Argon2.
As stated above, \cref{thm:main} only considers one parallel work $W$, but it naturally extends to multiple resources of each type.
In particular, our results capture weight functions such as $\Gamma(W_1, \ldots,W_k)$ with Condition~\ref{enum:thm:main:cond2} being $\alpha\Gamma(W_1, \ldots, W_k) = \Gamma(\alpha W_1, \ldots, \alpha W_k)$ for $\alpha > 0$.

Prior work~\cite{FWKKLVW22} suggested the weight function $\Gamma(W_1, \ldots, W_k) = \sum_{i=1}^k \omega_iW_i$ with constants $\omega_i$.
By \cref{thm:main}, this is secure against private double-spending, but not a desirable weight function in practice.
Even though the constants $\omega_i$ can be used to calibrate the contribution of each PoW, it seems difficult to realize this in a way that would prevent miners to ultimately only invest in the cheapest PoW.

Our result show that more interesting combinations of $W_1, \ldots, W_k$ are possible.
The first draws inspiration from automated-market-makers\footnote{\url{https://en.wikipedia.org/wiki/Constant_function_market_maker}} and is defined as
\begin{equation*}
        \Gamma(W_1, \ldots, W_k) = \Pi_{i = 1}^k W_i^{1/k}.
\end{equation*}
To maximize this weight function (for a given budget), one would have to invest into mining hardware for all PoWs at a similar rate.

Another option is the Leontief utilities function\footnote{\url{https://en.wikipedia.org/wiki/Leontief_utilities}}
\begin{equation*}
        \Gamma(W_1, \ldots, W_k) = \min\set*{W_1, \ldots, W_k}
\end{equation*}
which ensures that all PoWs must significantly contribute. 

Our work just classifies the weight functions that are secure in the sense that we get security against private double-spending whenever the honest parties control resources of higher weight (as specified by the weight function). A question that is mostly orthogonal to this work is to investigate which of those weight functions are also interesting from a practical perspective, say because they incentivize decentralization or other desirable properties. Let us observe that the class of secure weight functions does contain functions that make little sense in practice, for example the function $\Gamma(V)=V$ which simply counts the number of VDF steps. A blockchain based on this weight function would be secure assuming some honest party holds a VDF that is faster than the VDF held by the adversary.

\subsection{Model and Modelling Rationale}\label{S:model}
\subsubsection{Modelling Resources}
Our model captures resources that are \emph{external} to the chain, i.e., \emph{physical resources}.
In particular, we consider disk space $S$, sequential computation $V$, and parallelizable computation $W$ where we allow multiple resources per type, e.g., $W_1$ and $W_2$.\footnote{
        These are three fundamental resources in computation, and also the most popular physical resources used for blockchains.
        Nevertheless, we believe our model could be extended to other external resources.
}
Each resource is modelled as a function mapping time $\RRgez$ to an amount $\RRgz$.
This is expressive enough to capture, e.g., Bitcoin, any other PoW-based blockchain, or Chia.

In practice, cryptographic primitives are used to track these resources, usually \emph{Proof-of-Space} (PoSpace)~\cite{C:DFKP15}, \emph{Verifiable Delay Functions} (VDFs)~\cite{C:BBBF18,EC:Wesolowski19,ITCS:Pietrzak19b}, and \emph{Proof-of-Work} (PoW).
Our modelling essentially assumes a perfect primitive, glossing over implementation details and any probabilistic nature of the resource (similar to~\cite{terner22}). 

In practice parallel work $W$ as captured by a PoW, and sequential work $V$ as captured by a VDF are very different. $W$ is a quantitative resource in the sense that one can double it by investing twice as much, while $V$ is a qualitative resource as it measures the speed of the fastest available VDF. From the perspective of our Theorem on the other hand, $W$ and $V$ behave the same, the only thing that matters in the (proof of the) Theorem is that $W$ and $V$ are ``timed'' resources in the sense that their unit is something ``per second''. $W$ and $S$ on the other hand are both quantitative resources, but behave very differently.

\subsubsection{Reasons for Omitting Stake}
First, as described by Roughgarden and Lewis-Pye~\cite{permissionlessconsensus}, stake-based blockchains do not operate in the fully-permissionless setting.
Therefore, since our result targets this setting, modelling stake is not possible.

Another reason is that in our modeling we assume that parties hold some resources at some given time, for an on-chain resource like stake this is not well defined as the resource is only defined with respect to some particular chain, i.e., for different forks the parties would hold different resources at the same time. To make issues even more tricky, with stake it's possible to obtain old keys that no longer hold any value, but still can be used in an attack~\cite{ACLVZ23}.

\subsubsection{The Continuous Chain Model}
Towards \cref{thm:main}, we will first consider the \emph{Continuous Chain Model}.
While it is a very abstract model, it is rich enough to already yield Conditions~\ref{enum:thm:main:cond1}~and~\ref{enum:thm:main:cond2}.
In a nutshell, we assume that a chain \emph{continuously} and \emph{exactly} reflects the resources that were expended to create it.

Assume that the honest parties at time $t$ hold resources $S(t),W(t),V(t)$, then the chain they can create in a time window $[t_0,t_1]$ will have weight 
$\int_{t_0}^{t_1} \Gamma(S(t), V(t), W(t))\,dt$.

This continuous model avoids some issues of actual blockchains, like their probabilistic nature or network delays. 
For example in Bitcoin, a so called $51\%$ attack can actually be conducted with less, say, $41\%$ of the hashing power if network delays are sufficiently large. 
Moreover, as a Bitcoin block is found every 10 minutes \emph{in expectation}, it frequently happens that no block is found in an hour at all. For this reason a block is only considered confirmed if it's sufficiently deep in the chain. These factors only have a \emph{quantitative} impact on the concrete security threshold of longest-chain 
blockchains and are well understood~\cite{CCS:GazKiaRus20,CCS:DKTTVW20}. 
The goal of this paper, however, is \emph{qualitative} in nature.
That is, we want to describe which weight functions are secure against private double-spending attacks as long as honest parties have sufficiently more resources than the adversary.
Precisely quantifying how much security is lost due to the fact that resources are only approximately recorded,  due to network delays or other 
aspects like double dipping attacks is not our goal.

So far we considered a strongly idealized setting where the blockchain recorded the available resources \emph{continuously} and \emph{exactly}, both can not be met in a practical blockchain\footnote{At least not if they use a quantitative resource $S$ or $W$, which only leaves $V$, but speed alone will not make a good chain.} where the quantitative resource $S$ or $W$ is distributed over an a priori unlimited number of miners, but for practical reasons we only want a bounded number to actually give input to a block. 
In Bitcoin and Chia, it is just a single miner that finds a proof that passes some difficulty, and the frequency at which such proofs are found gives an indication of the total resource.
We can get a good approximation of the resources by waiting for sufficiently many blocks or using a design where multiple miners contribute to a block, say we record some $k>1$ best proofs found since the last block in every block. In this work, we will not further deal with the fact that resources are not \emph{exactly} recorded as it is not very informative for the ideal perspective we are taking. In actual constructions like Bitcoin one deals with this by requiring some time before considering blocks as confirmed. The number of blocks to wait is computed using a tail inequality, it depends on the probability of failure one can accept and on the quantitative gap one assumes between honest and adversarial resources.

\subsubsection{Private Double-Spending Attack (PDS)}\label{sec:pds}\label{S:PDSonly}
\paragraph{Why We Focus on PDS}
We analyze the security of weight functions against a specific attack, the \emph{Private Double-Spending} (PDS) attack.
So we do not prove security against arbitrary attacks and cannot rule out that a worse attack than PDS exists.

The works of Dembo~et~al.~\cite{CCS:DKTTVW20} and Gaži~et~al.~\cite{CCS:GazKiaRus20} used different techniques to show that PDS is the worst attack against PoW and Proof-of-Stake-based longest-chain protocols.
The former~\cite{CCS:DKTTVW20} analysis also extends to Chia.
Concretely, they show that if an adversary has sufficient resources to perform some attack against one of these blockchains, they could also perform a private double-spending attack instead. 
This verifies the intuition of Nakamoto, who only considered the double-spending attack  when arguing about Bitcoin's security~\cite{bitcoin}.

The above results~\cite{CCS:DKTTVW20,CCS:GazKiaRus20} are not general enough to imply the same (i.e., that only consider PDS attacks is sufficient) in our more general setting.
For example, their analyses do not capture a blockchain design relying on two different PoWs---a design which our model allows.
Nevertheless, these works give evidence that focusing on PDS attacks is enough, and we believe that the results of \cite{CCS:DKTTVW20} should generalize to our setting. We leave proving or refuting this intuition to future work.

Note that our intuition is based on the technical details of Dembo~et~al.'s analysis~\cite{CCS:DKTTVW20}.
It relies on a connection between PDS and any general attack strategy. 
That is, one \enquote{can view \emph{any} attack as a race between adversary and honest chains, not just the private attack. However, unlike the private attack, a general attack may send many adversary chains to simultaneously race with the honest chain.}~\cite[p. 2]{CCS:DKTTVW20}. 

\paragraph{Explanation of PDS}
In a PDS attack, the adversary forks the chain at some point in time, privately extends its own fork (while honest parties continue to extend the main chain), and releases its fork later on.
The attack is successful if the adversary's fork is at least as heavy as the honest chain since this would allow the adversary to double-spend a transaction.

We let the adversary choose the resources available to honest parties and itself during this time.
The only condition is that we disallow the adversary to trivially perform a successful attack.
That is, at every point in time the adversary resources have at most as much weight as the resources of honest parties, and, to avoid a draw, in some interval strictly less.

The honest chain directly corresponds to the resources of the honest parties.
The adversary, however, may cheat since it is mining in private.
In particular, it can pretend to have created the chain in a shorter or longer amount of time by stretching/squeezing time.
This time manipulation affects the resources recorded on the chain.
For example, consider $W$ as the hash rate, then a chain records the total number of hashes.
If the adversary now pretends to have created this chain in $1/2$ time, then its hashrate must be $2\cdot W$ since the total amount of hashes does not change.

Given such an attack, e.g., the weight function $\Gamma(W) = W^2$ is insecure.
Indeed, consider an adversary with resource $W(t) = 1$ and honest parties with $W(t) = 2$.
Honest parties mining for $1$ time create a chain of weight $1\cdot2^2 = 4$.
In the same timespan, the adversary creates a chain of weight $1\cdot1^2 = 1$.
However, if it pretends to have mined this chain privately in $1/8$ time, then its chain records the weight $1/8\cdot (8\cdot1)^2 = 8$ instead, beating the honest chain.

We defer more precise definitions and figures exemplifying this time manipulation to~\cref{sec:ideal}.

\subsubsection{The Discrete Chain Model}
So far, we discussed an abstract model where the chain \emph{continuously} and \emph{exactly} records resource expenditure. 

In \S\ref{sec:discrete} we discuss a model closer to a real blockchain, where blocks arrive in discrete time slots. 
We still assume the block exactly records the resources $W$ and $V$. In particular, a block produced during some timespan $[a,b)$ records 
$W_\blockk = \int_a^{b} W(t)\,dt$ and $V_\blockk = \int_a^{b} V(t)\,dt$. 
For the space we assume that the block records $S(t)$ at some point $a \leq t < b$.
The reason for this difference is that $W$ and $V$ are resources that are measured \emph{per second} (e.g., hashes/s or steps/s), so integration over time is well-defined.
One the other hand, a proof of space gives a snapshot of the space $S(t)$ available at some point $t$ during block creation. To be on the safe side, we simply assume that 
the adversary can choose the time $t$ where its space was maximal, while for the honest parties we assume $t$ is the time when $S(t)$ was minimal.

We show (Theorem~\ref{thm:discrete}) that the classification of secure weight functions basically carries over to this discrete setting as long as the resources don't vary by too much within the block arrival time.  But our main motivation to consider the discrete model is to discuss the issue of replotting attacks in \S\ref{S:replot}, which only make sense in a discrete setting.

\subsection{Future Work}
Our work opens multiple new questions for future work. 
We already mentioned identifying weight functions that are not only secure, but also interesting at the end of section  \S\ref{sec:intro:impl}. At the end of section \S\ref{S:replot} we will discuss an open question concerning replotting attacks. Some other open questions include:

First, modelling on-chain resources, most notably, stake.
While stake somewhat behaves like disk space\footnote{There exist proposals similar to Chia that use stake instead of disk space, i.e., where the weight is $\text{Stake} \cdot V$~\cite{FC:DebKanTse21}.}, it is different and difficult to model since it is an on-chain resource.
For example, one modelling challenge is capturing long range attacks in which parties sell old keys that controlled a lot of stake at some point.
This is similar to a bootstraping attack for disk space, but the difference is that the adversary can perform this attack for free (after having bought the keys).

Second, considering chain-selection rules other than the heaviest-chain rule.
For example, in the stake setting, Ourboros Genesis~\cite{genesis} operating in dynamically available setting achieves security against PDS using a different chain selection rule.

Third, considering different degrees of permissionless, such as the dynamically-available or quasi-permissionless setting described by~\cite{permissionlessconsensus}.
While our results rule out solely using disk space in the fully-permissionless setting, this impossibility does not carry over to other models.
For example, Filecoin~\cite{filecoin} only uses disk space, but is secure against PDS because it operates in the quasi-permissionless model.

\subsection{Related Work}
\subsubsection{Abstract Resource Models}
Recently, Roughgarden~and~Lewis-Pye~\cite{permissionlessconsensus} (an updated version of~\cite{FC:LewRou23}) presented many (im)possibility results about permissionless consensus.
They consider a resource-restricted adversary where resources can be external or on-chain resources.
External resources are modelled by so-called \emph{permitter oracles}, whose outputs depend on the amount of resources the querying party has at the time of the query.
An important part of their work is a classification of the permissionlessness of consensus protocols:
\begin{itemize}
        \item \emph{Fully-permissionless} protocols are oblivious to its participants (e.g., Bitcoin).
        \item \emph{Dynamically-available} protocols know a dynamic list of parties, which may be a function of the past protocol execution (e.g., parties who staked coins), the participants are a subset of this list, and \emph{at least one} honest member of this list participates.
        \item \emph{Quasi-permissionless} protocols are similar to dynamically-available protocols, but make the stronger requirement that \emph{all} honest members of the list participate.
                Note that such protocols differ from \emph{permissioned} ones, which also have a list of parties, but where the list \emph{cannot} depend on the past protocol execution.
\end{itemize}
For example, a result of theirs states that no deterministic protocol solves Byzantine agreement in the fully-permissionless setting, even with resource restrictions.

Two preceding works modelling abstract resources are Terner~\cite{terner22} and Azouvi et~al.~\cite{ACLVZ23}.
Terner~\cite{terner22} considers an abstract resource that essentially is a black-box governing participant selection.
They give a consensus protocol that can be instantiated with any such resource satisfying certain properties (e.g., resource generation must be rate-limited relative to the maximum message delay). Both \cite{terner22,ACLVZ23} consider only a single resource and not combination of multiple resources.   

Azouvi~et~al.~\cite{ACLVZ23} use the abstraction of resource allocators (similar to permitter oracles in~\cite{permissionlessconsensus}) to build a total-ordered broadcast protocol.
They describe the properties a resource allocator must fulfill (e.g., honest majority), and construct resource allocators for the resources stake, space\footnote{Their space allocator lies in the quasi-permissionless model, thus not conflicting with our results.}, and work.
As part of this, they classify resources as external vs.\ on-chain (they call it virtual), and burnable vs.\ reusable (space and stake are reusable whereas work is not) and discuss trade-offs between different types of resources, e.g., on-chain resources are susceptible to long-range attacks. A limitation of~\cite{ACLVZ23} is that total resource is \emph{a priori} known and fixed. In our model all resources can vary and are known only when the blocks are created.

\subsubsection{Blockchain Designs}
We give a selection of well-known permissionless blockchain designs, describing the weight function or---for non-Nakamoto-like protocols---resource ($S$, $V$ and $W$ as before, and stake $\stake$) and degree of permissionlessness (\textbf{f}ully-\textbf{p}ermissionless, \textbf{d}ynamically-\textbf{a}vailable, or \textbf{q}uasi-\textbf{p}ermissionless):
Bitcoin~\cite{bitcoin} (fp, $W$),
Chia~\cite{chia} (fp, $S \cdot V$),
Filecoin~\cite{filecoin} (qp, $S$),
Ethereum~\cite{ethereum} (da/qp,\footnote{Depending on whether the network is synchronous/partially-synchronous~\cite{permissionlessconsensus}.} $\stake$),
Algorand~\cite{algorand} (qp, $\stake$),
Ouroboros~\cite{C:KRDO17,EC:DGKR18,CCS:BGKRZ18} (da/qp, $\stake$),
Snow White~\cite{FC:DaiPasShi19} (da, $\stake$).

Multiple combinations of proof-of-stake (PoStake) and PoW, e.g.,~\cite{King2012PPCoinPC,BLMR14,FWKKLVW22} exist (all either da or qp).
\cite{FWKKLVW22} is the only \emph{fungible} protocol, i.e., it is secure as long as the adversary controls less than half of all stake and work \emph{cumulatively} (essentially mapping to $\Gamma(\stake, W) = \stake + W$). 
Their protocol handles multiple PoStake and multiple PoW resources, thus capturing $\Gamma(W_1, W_2) = W_1 + W_2$, which we discussed in \cref{sec:intro:impl}.

\cite{FC:DebKanTse21} combine PoStake with sequential computation to create a dynamically-available protocol.

Ignoring difficulty, Bitcoin assigns unit weight to every block whose hash surpasses a threshold.
\cite{KMMNTT21} analyze other functions assigning weight to block hashes.
They suggest using a function that grows exponentially, but is capped at a certain value, which depends on the maximum network delay.

\subsubsection{Analyses of Blockchain Protocols}
Various works analyze specific blockchains, mostly Bitcoin (or similar PoW chains)~\cite{EC:GarKiaLeo15,EC:PasSeeshe17,CSF:PasShi17,EPRINT:Ren19,CCS:GazKiaRus20}, but also longest-chain protocols in general~\cite{CCS:DKTTVW20}.
These generally give quantitative security thresholds (i.e., what fraction of adversarial resources is tolerable) depending on, e.g., maximum message delay.
We again remark that our work has a different aim, namely a qualitative description of weight functions, disregarding precise security thresholds.

\section{Preliminaries} 
Let $[n] = \set{1, \ldots, n}$.
Vectors are typeset as bold-face, e.g., $\bm{x}$.
$\RRgz$ and $\RRgez$ denote the set of positive real numbers excluding and including $0$, respectively. 
Given two tuples $(x_1, \ldots, x_n), (x'_1, \ldots, x'_n) \in \RRgz^n$ we say $(x_1, \ldots, x_n) \leq (x'_1, \ldots, x'_n)$ if $x_i \leq x'_i$ for all $i \in [n]$ with equality holding if and only if $x_i = x'_i$ for all $i \in [n]$. 

We denote the time by $t \in \RRgez$.
For $T_0, T_1 \in \RRgez$ where $T_0<T_1$, $[T_0, T_1]$ denotes the time interval starting at $T_0$ and ending at $T_1$.
The open interval $(T_0,T_1]$ denotes the time interval $[T_0,T_1]$ excluding $T_0$.
$[T_0,T_1)$ is defined analogously.

\begin{definition}[Monotonicity]\label{def:monotone}
	A function $f\colon \RRgz^n \to \RRgz$ is monotonically increasing if 
	\begin{equation*}
		(x_1, \ldots, x_n) \leq (x'_1, \ldots, x'_n) \implies f(x_1, \ldots, x_n) \leq f(x'_1, \ldots, x'_n). 
	\end{equation*}
\end{definition}

\begin{definition}[Homogeneity]\label{def:homogenous}
        A function $f\colon \RRgz^n \to \RRgz$ is homogeneous\footnote{More precisely, $f$ is a positively homogeneous function of degree $1$. However, we will not need homogeneity of higher degree, so we simply call it \enquote{homogeneous}.} in $x_j, \ldots, x_n$ with $1 \leq j \leq n$ if,
        for all $(x_1, \ldots, x_n) \in \RRgz^n$ and $\alpha > 0$,
        \begin{equation*}
                f( x_1, \ldots, x_{j-1}, \alpha \cdot x_j, \ldots, \alpha \cdot x_n) 
                =
                \alpha\cdot f( x_1, \ldots, x_{j-1},  x_j, \ldots, x_n) 
        \end{equation*}
\end{definition}

\begin{definition}[Subhomogeneity]\label{def:subhomo}
        A function $f\colon \RRgz^n \to \RRgz$ is subhomogeneous in $x_1, \ldots, x_j$ with $1 \leq j \leq n$ if, for all $(x_1, \ldots, x_n) \in \RRgz^n$ and $\alpha \geq 1$,
        \begin{equation*}
                f(
                             \alpha \cdot x_1, \ldots, \alpha \cdot x_{j}, 
                             x_{j+1}, \ldots, x_n
                 ) 
                 \leq 
                 \alpha \cdot f(
                             x_1, \ldots, x_{j}, 
                             x_{j+1}, \ldots, x_n
                 ).
        \end{equation*}
\end{definition}

\section{Continuous Chain Model}\label{sec:ideal}
To characterize which weight functions provide security against private double-spending (PDS) attacks, we will first introduce the \emph{Continuous Chain Model}.
It models physical resources, how resources are turned into an idealized blockchain, and PDS attacks.
As the name suggests, the continuous model views the blockchain as one continuous object, instead of consisting of multiple discrete blocks.

\subsection{Modelling Resources}
The model captures physical resources, which are external to the chain, and allows for multiple resources per type.
The resources are disk space $\Sb \eqq (S_1, \ldots, S_{k_1})$, sequential work $\Vb \eqq (V_1, \ldots, V_{k_2})$, and parallel work $\Wb \eqq (W_1, \ldots, W_{k_3})$.\footnote{
        These are three fundamental resources in computation and also the most popular physical resources used for blockchains.
        Nevertheless, we believe our model could be extended to other external resources.
}
It allows for multiple resources per type.
We will omit $k_1$, $k_2$, and $k_3$ unless needed for clarity.

A \emph{resource profile} records the amount of each resource available at any point in time.
Time is modelled as a continuous variable $t \in \RRgez$, and we restrict our attention to the time interval $[0, T]$ for some $T > 0$.
We take each resource to be a function mapping this interval to $\RRgz$, e.g., $W_1 \colon [0, T] \to \RRgz$.

\begin{definition}[Resource Profile]\label{def:resProf}
	A resource profile $\Rc$ is a 3-tuple of tuple of functions 
    \begin{equation*}
            \Rc \eqq \left( \Sb(t), \Vb(t), \Wb(t) \right)_{[0,T]}
    \end{equation*}
     where each tuples of functions is composed of functions with domain $t \in [0,T]$ with $T > 0$ and range $\RRgz$, and where each function is Lebesgue integrable. 
     
\end{definition}

\begin{remark}
        The requirement that each resource is non-zero at every point in time is a minor technical condition.
        Note that it is always fulfilled in practice since interaction with the blockchain requires a general-purpose computer, and even a low-powered one provides a non-zero amount of $S$, $V$, and $W$.
\end{remark}

\subsection{Idealized Chain}
Ideally, a blockchain should record the amount of resources that were expended to create it, and blockchain protocols are generally designed to approximate this as closely as possible.
In our idealized model, we assume that a blockchain \emph{continuously} and \emph{exactly} reflects the resources expended to create it.

\begin{definition}[Continuous Chain Profile]\label{def:chainProf}
	A continuous chain profile $\chain$ is a 3-tuple of tuple of functions 
    \begin{equation*}
            \chain \eqq \left( \Sb(t), \Vb(t), \Wb(t) \right)_{[0,T]}
    \end{equation*}
    where each tuples of functions is composed of functions with domain $t \in [0,T]$ where $T > 0$ and range $\RRgz$, and where each function is Lebesgue integrable.
\end{definition}

\begin{remark}
        Resource and chain profiles are syntactically identical.
        The difference lies in semantics:
        A resource profile describes the resources available to a party (or a set of parties).
        Meanwhile, a chain profile describes the resources that the chain reflects.
\end{remark}

\begin{remark}
        In practice, blockchains do not \emph{exactly} record the amount of resources, but only approximate them.
        For example, in Bitcoin, finding blocks is a probabilistic process, so blocks do not record the actual work invested to create them, but only the \emph{expected} amount of work.
        Additionally, network delays cause miners to waste time (and thereby work) trying to extend an out-of-date block, in the worst case leading to orphaned blocks.
        In spite of these issues, the ideal model is still meaningful because these issues introduce \emph{quantitative} gaps (e.g.,~\cite{CCS:DKTTVW20,EC:GarKiaLeo15,CCS:GazKiaRus20}).
\end{remark}

To capture the heaviest-chain rule, the model assigns each chain a weight.
To this end, we first introduce the \emph{weight function} $\weight$, which assigns a weight to a triple of resources.
In other words, it assigns a weight to one point in time.

\begin{definition}[Weight Function]\label{def:weightFunc}
        A weight function is a non-constant function given by 
        \begin{equation*}
                \weight\colon \RRgz^{k_1} \times \RRgz^{k_2} \times \RRgz^{k_3} \to \RRgz.
        \end{equation*}
\end{definition}

\begin{remark}
        The requirement that $\weight$ is non-constant is natural in practice.
        We explicitly require it because such functions would be vacuously secure against PDS.
        Looking ahead, the security definition only considers adversaries with a resource disadvantage, which is measured using weight.
        But if the weight is constant, no such adversary exists, so the function would always be secure against PDS.
\end{remark}

As said, $\weight$ takes in three resources and outputs the weight of a specific point in time.
In the next step, we extend $\weight$ to compute the weight of a whole chain.
We denote this function by $\cweight$.
It takes a continuous chain profile as input and outputs the weight of it.
Overloading notation slightly, we also allow inputting a resource profile to $\cweight$ since it is syntactically identical to a chain profile.

\begin{definition}[Weight of a Chain or Resource Profile]\label{def:chain_quality}
        Consider a weight function $\weight$ and a continuous chain $\chain = (\Sb(t), \Vb(t), \Wb(t))_{[0, T]}$ or resource profile $\Rc = (\Sb'(t), \Vb'(t), \allowbreak\Wb'(t))_{[0,T]}$.
        The chain weight function $\cweight$ is defined as 
        \begin{align*}
                \cweight(\chain) \eqq \int_{0}^{T} \weight(\Sb(t), \Vb(t), \Wb(t)) \,dt
                \shortintertext{and}
                \cweight({\Rc}) \coloneq \int_{0}^{T} \weight(\Sb'(t), \Vb'(t), \Wb'(t)) \,dt.
        \end{align*}
\end{definition}

\subsection{The Private Double-Spending Attack}
In a private double-spending (PDS) attack, the adversary forks the chain at some point in time, extends this fork in private, before releasing the private fork to the public.
The attack is successful if the adversary's fork is heavier than the honest chain, because the adversarial fork replaces the honest chain, effectively reverting past transactions.
We focus on the PDS attack because it is the prototypical attack against blockchains; we refer back to \cref{sec:pds} for an in-depth discussion.

\subsubsection{Modelling the Attack}
To model this attack, we first consider the time frame of the attack.
The attack starts (i.e., the adversary forks the chain) at time $0$, and the adversary publicly publishes its private chain at time $T_\eend$.
So the attack spans the time interval $[0, T_\eend]$.
During this time, the resources available to the honest parties are given by the resource profile $\Rc^\Hc$, and they use them to build the honest chain profile $\chain^\Hc$ in the following way:

\begin{definition}[Honest Chain Profile]
	    Consider a resource profile $\Rc^\Hc = (\bhsp(t), \bhvdf(t), \allowbreak\bhw(t))_{[0,T_\eend]}$ 
        The corresponding honest chain profile is $\chain^\Hc \eqq \Rc^\Hc$.
\end{definition}

That is, the honest chain $\chain^\Hc$ correctly reflects the resources available to honest parties, and also precisely keeps track at which point in the resources were available.\footnote{
        While this is by construction in our idealized model, timestamps are generally accurate in longest-chain blockchains---even if a not-too-powerful adversary tries to disrupt them~\cite{EPRINT:TziSriZin23}.
}

The adversary's resources are $\Rc^\Ac$, and they use them to build the chain profile $\chain^\Ac$.\footnote{In general, we use the superscripts $\phantom{\cdot}^\Hc$ and $\phantom{\cdot}^\Ac$ to denote the honest parties and the adversary, respectively.}
In contrast to the honest parties, the adversary may deviate from the protocol and cheat.
First, they may simply not use some of the resources available to them.
Second, and more importantly, since the adversary creates the fork in private, the chain $\achain$ may not correctly reflect \emph{at what time} the resources were available.
In essence, the adversary can \emph{alter} the time by \emph{stretching} and \emph{squeezeing} it.
For example, in Bitcoin the adversary may mine a block in $100$ minutes but pretend to have mined it within $10$ minutes.

We model this time manipulation by a function $\phi(t)$ describing the squeezing/stretching factor at any point in time. 
At time $t$, $\phi(t)>1$ represents squeezing, $\phi(t) <1$ stretching, and $\phi(t) = 1$ no alteration.
Altering the time affects, e.g., the length of the interval $[0, T_\eend]$.
To this end, we introduce the \emph{altered time} function $\at$ (and its inverse $\atinv$)\footnote{
        $\atinv$ exists because $\frac{1}{\phi(u)} > 0$ for all $u$, so $\int_{0}^{t} \frac{1}{\phi(u)} \,du$ is a monotonically increasing function of $t$ with co-domain $[0,\at(T_\eend)]$.
} to translate between time before and after squeezing/stretching.
For example, $\Ttil_\eend = \at(T_\eend)$, resulting in the interval $[0, \Ttil_\eend]$.\footnote{We use $\widetilde{\phantom{\circ}}$ to denote time after squeezing/stretching.}

Altering time affects how $\achain$, which covers the time interval $[0, \Ttil_\eend]$, reflects resources.
For the resources $V$ and $W$, altering time cannot change the cumulative amount (e.g., in Bitcoin it cannot change the number of found blocks and thus work performed).
Therefore, $V$ and $W$ must be multiplied by $\phi$.
That is, at altered time $\tilde{t} \in [0, \Ttil_\eend]$, $\achain$ records the resource $\phi(t)V^{\Ac}(t)$ and $\phi(t)W^{\Ac}(t)$.
The disk space $S$ behaves differently.
As long as it is available, it can be reused~\cite{ACLVZ23}, so it does not accumulate (unlike $V$ and $W$).
As a consequence, altering time just changes when space was available.
Thus, at altered time $\tilde{t} \in [0, \Ttil_\eend]$, $\achain$ records $S(\atinv(\tilde{t}))$.

\begin{definition}[Adversarial Chain Profile]\label{def:idealAttack}\label{D:advchainprofile}
	Consider a resource profile $\Rc^\Ac = (\basp(t), \bavdf(t), \allowbreak\baw(t))_{[0,T_\eend]} $ 
	and some function $\phi(t)\colon [0,T_\eend] \to \RRgz$.
    Define $\at(t) \coloneq \int_{0}^{t} \frac{1}{\phi(u)} \,du$ and its inverse $\atinv(\cdot)$.	

	Let $\Ttil_\eend \eqq \at(T_\eend)$.
	An adversarial chain profile is any chain profile 
    \begin{equation*}
            \achain = (\basptil(\ttil), \bavdftil(\ttil), \bawtil(\ttil))_{[0,\Ttil_\eend]}
    \end{equation*}
    where $\asptil_i(\cdot)$, $\avdftil_i(\cdot)$ and $\awtil_i(\cdot)$ are Lebesgue integrable, and satisfy
    \begin{align*}
            0 &< \basptil(\ttil) \le \basp(t)\\
            0 &< \bavdftil(\ttil) \le \phi(t) \cdot \bavdf(t)\\
            0 &< \bawtil(\ttil) \le \phi(t) \cdot \baw(t)
    \end{align*}
	for all $\ttil \in [0, \Ttil_\eend]$ with $t = \atinv(\ttil)$.
\end{definition}

Let us illustrate the stretching and squeezing from \cref{def:idealAttack} by the example of Bitcoin and Chia in \cref{fig:bitcoin,fig:chia}.
\begin{figure}[h]
    \centering
    \includegraphics[page=1,width=\textwidth]{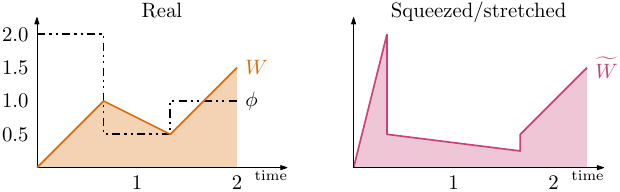}
    \caption{Bitcoin's weight function $W$ and how it reacts to stretching and squeezing. The shaded area is the weight.}\label{fig:bitcoin}
\end{figure}
\begin{figure}[h]
    \centering
    \includegraphics[page=2,width=\textwidth]{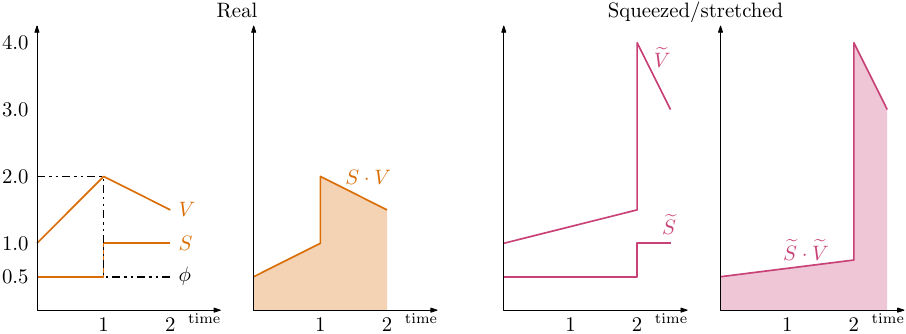}
    \caption{Chia's weight function $S\cdot V$ and how it reacts to stretching and squeezing. The shaded are is the weight.}\label{fig:chia}
\end{figure}

We now have all ingredients to define when a weight function is secure against PDS attacks.
On a high level, the definition states that an adversary having resources of less weight than the honest parties\footnote{
Clearly, a PDS attack always works when the adversary has a resource advantage.
} cannot create a private chain that is heavier than the honest parties one---even by manipulating time.
In more detail, \enquote{less weight} means that the adversary has at most equal weight at every point in time (\cref{eq:secure:1}), and in some interval it has strictly less (\cref{eq:secure:2}).

\begin{definition}[Weight Function Security Against PDS, Continuous Model]\label{def:secure}
        \sloppy A weight function $\weight$ is \emph{secure against private double-spending attack} in the \emph{continuous model} if for all  $\Rc^\Hc = (\bhsp(t), \bhvdf(t), \bhw(t))_{[0,T_\eend]}$ and
        $\Rc^\Ac = (\basp(t), \bavdf(t), \baw(t))_{[0,T_\eend]}$ 
        such that 
        \begin{equation}
                \weight(\basp(t), \bavdf(t), \baw(t)) \leq \weight(\bhsp(t), \bhvdf(t), \bhw(t)) \: \forall t \in [0, T_\eend]\label{eq:secure:1}
        \end{equation}
        and for a time interval $[T_0, T_1]$ 
        \begin{equation}
                \weight(\basp(t), \bavdf(t), \baw(t)) < \weight(\bhsp(t), \bhvdf(t), \bhw(t)) \: \forall t \in [T_0, T_1]\label{eq:secure:2}
        \end{equation}
        it holds that
        \begin{equation*}
                \cweight(\hchain) > \cweight(\achain)
        \end{equation*}
        where $\hchain \eqq \Rc^\Hc$ and $\achain$ satisfies \cref{def:idealAttack} for $\Rc^\Ac$ and any $\phi(t)$.
\end{definition}

\begin{remark}
        An alternative to the precondition (\cref{eq:secure:1,eq:secure:2}) on resource profiles in~\cref{def:secure} is that adversarial resources must be strictly smaller than the honest ones at every point in time (instead of just for an interval).
        Looking ahead, \cref{thm:idealModSec} would be true in the if-direction (monotonically increasing and homogeneous implies secure against PDS), but not in the only-if direction.
        The reason is that not every non-homogeneous function can be attacked (e.g., a function that is $S \cdot \max(V,W)$
        when each resource is below some constant threshold $c$ and that is constant $c^2$ after that).
        If we additionally put the natural constraint that a weight function $\Gamma$ is not eventually constant (i.e., for any point $(\ssb, \vb, \wb)$ there exists $(\ssb',\vb',\wb')$ such that $(\ssb, \vb, \wb) < (\ssb', \vb', \wb')$ and $\Gamma(\ssb, \vb, \wb) < \Gamma(\ssb', \vb', \wb')$), then only-if direction also holds (by an adaption of our proof).
        Either way, the main takeaway is that monotonically increasing and homogeneous functions are the ones secure against PDS, and they are the only ones that should be used to construct Nakamoto-like blockchains using multiple resources.
\end{remark}

\subsection{Main Theorem in the Continuous Model}
There are many possible choices for $\Gamma$, but not all are secure against PDS, i.e., a bad choice for a blockchain.
For example, \cref{fig:PoWSS} shows that $W_1 \cdot W_2$ is insecure, but that $\sqrt{W_1} \cdot \sqrt{W_2}$ seems secure---at least in the example.
The following theorem shows that it is secure against PDS in general, because it is monotone (cf.~\cref{def:monotone}) and homogeneous in $\Wb$ and $\Vb$ (i.e., $\alpha\Gamma(\Sb, \Vb, \Wb) = \Gamma(\Sb, \alpha\Vb, \alpha\Wb)$, cf.~\cref{def:homogenous}). 
These are sufficient, but also necessary conditions for security against PDS in the continuous chain model.

\begin{figure}[h]
        \centering
        \includegraphics[width=\textwidth,page=1]{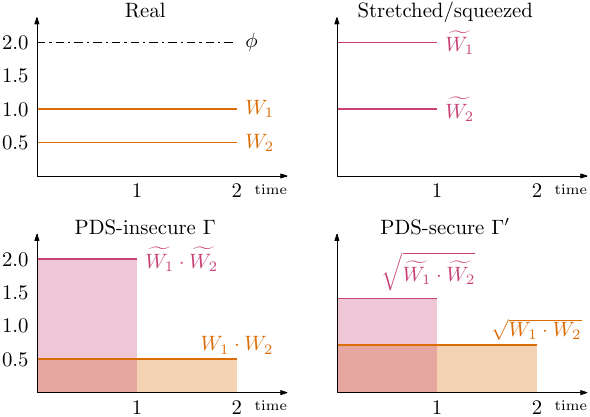}
        \caption{\label{fig:PoWSS}
                Consider two PoWs $W_1, W_2$, and two weight functions $\Gamma(W_1, W_2) = W_1 \cdot W_2$ and $\Gamma'(W_1, W_2) = \sqrt{W_1 \cdot W_2}$.
                The top row show the real resources $\textcolor{myorange}{W_1, W_2}$ (left) and how squeezing them by $\phi(\cdot) = 2$ (left) results in \textcolor{mypink}{$\til{W_1}, \til{W_2}$} (right).
                The bottom row shows that $\Gamma$ is not secure because $\textcolor{myorange}{\int_0^2 W\cdot V} < \textcolor{mypink}{\int_0^1 \widetilde{W} \cdot \widetilde{V}}$, i.e., squeezing increases the weight.
                In contrast, $\Gamma'$ is not affected by the squeezing.
        }
\end{figure}

\begin{theorem}[Secure Weight Functions, Continuous Model]\label{thm:idealModSec}
        A weight function $\weight$ is \emph{secure against private double-spending attacks} in the \emph{continuous model} if and only if $\weight(\Sb,\Vb,\Wb)$ is monotonically increasing (\cref{def:monotone}) and homogeneous in $\Vb,\Wb$ (\cref{def:homogenous}).
\end{theorem}

 We defer the proof to \cref{sec:idealProof}.

\section{Discrete Chain Model}\label{sec:discrete}
The continuous chain model is rather abstract, so we also consider a discretized version using blocks.
A block reflects the \emph{total} amount of resources that were expended to create it.
Honest users create blocks according to some prescribed rule, e.g., in fixed time intervals, but the adversary may not adhere to this rule.

Like in the continuous model, we will describe which weight function $\weight$ leads to a discrete blockchain that is secure against PDS\@.
Compared to the continuous model the security statement introduces quantitative factors.  %
The reason is that resources can fluctuate within a block.
The quantitative parameters essentially state: The higher the magnitude of fluctuations within blocks, the larger the resource disadvantage of the adversary must be.

In principle, this statement also needs to quantitatively depend on how $\weight$ depends on $\Sb$.
The reason is that in our modelling a block reflects $\Sb$ available \emph{at some point in time} during the block creation.
Since $\Sb$ can fluctuate within a block, we pessimistically assume that the adversary always gets the maximum and the honest parties only the minimum.
So, e.g., the function $S^2\cdot V$ requires a larger disadvantage than $S \cdot V$.
To not carry around another parameter, we restrict our attention to natural choices for $\weight$, namely, $\Gamma$ that are subhomogeneous in $\Sb$ (cf.~\cref{def:subhomo})---think of this as \enquote{at most linear in $\Sb$}.

\subsection{Definitions}
We define a blockchain $\bchain$ as the discretization of a resource profile $\res$.
Let us first define a \emph{block}.

\begin{definition}[Blocks]\label{def:block}
	Let $\res = (\Sb(t), \Vb(t), \Wb(t))_{[0,T]}$ be a resource profile.
	A block $b_i$ is defined by a timespan $(t_i, t'_i)$ with $0 \leq t_i < t'_i \leq T$.
	The resources reflected by the block are denoted by $\bSb(b_i)$, $\bVb(b_i)$, and $\bWb(b_i)$.
	
	Timed resources $\bVb$ and $\bWb$ are reflected by 
	\begin{equation*}
		\bVb(b_i) = \int_{t_i}^{t'_i} \Vb(t) \,dt \quad\text{and}\quad \bWb(b_i) = \int_{t_i}^{t'_i} \Wb(t) \,dt.
	\end{equation*}
	The constraint on $\bSb$ is that
	\begin{equation}\label{eq:disc:space-reflected}
		\inf_{t_i < t < t'_i} \Sb(t) \leq \bSb(b_i) < \sup_{t_i < t < t'_i} \Sb(t).
	\end{equation}
	The weight of a block $b$ is $\weight(\bSb(b), \bVb(b), \bWb(b))$.
\end{definition}

The resources $\Vb$ and $\Wb$ accumulate within a block (e.g., a Bitcoin block reflects the expected number of hashes).
$\Sb$ is different (\cite{ACLVZ23} called it \enquote{reusable}), so a block can only reflect \emph{some} amount of space that was available within the block's timespan.

In the sequel, we will often make use of minima and maxima of resources within a block.
For technical reasons, they are defined via infimum and supremum, but think of them as minimum and maximum.

\begin{definition}[Minima and Maxima of Resources]
        For a resource profile $\res = (\Sb(t), \Vb(t),\allowbreak \Wb(t))_{[0,T]}$ we denote the minima/maxima of resources in a block $b$ with timespan $(t',t'')$ by
        \begin{align*}
                \bSbmin(b) &= \inf_{t' < t < t''} \Sb(t)&\bSbmax(b) &= \sup_{t' < t < t''} \Sb(t)\\
                \bVbmin(b) &= \inf_{t' < t < t''} \Vb(t)&\bVbmax(b) &= \sup_{t' < t < t''} \Vb(t)\\
                \bWbmin(b) &= \inf_{t' < t < t''} \Wb(t)&\bWbmax(b) &= \sup_{t' < t < t''} \Wb(t)
        \end{align*} 
        where $\inf$ and $\sup$ are applied element-wise over the whole vector.
\end{definition}

Now, a blockchain is a sequence of non-overlapping blocks.
Its weight $\bcweight$ is the sum of the blocks' weights.

\begin{definition}[Discrete Blockchain]\label{def:disc}
	A discrete blockchain is a sequence of blocks $\bchain = (b_0, \ldots b_B)$ whose timespans do not overlap. 
	The weight of a blockchain is
	\begin{equation}
		\bcweight(\bchain) = \sum_{b_i \in \bchain} \weight(\bSb(b_i), \bVb(b_i), \bWb(b_i)) 
	\end{equation}
\end{definition}

\begin{figure}
	\centering
	\begin{subfigure}[b]{0.45\textwidth}
		\centering
		\includegraphics[width=\textwidth]{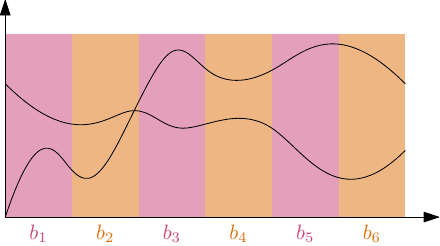}
		\caption{Honest parties.}\label{fig:chunkHon}
	\end{subfigure}
	\hfill
	\begin{subfigure}[b]{0.45\textwidth}
		\centering
		\includegraphics[width=\textwidth]{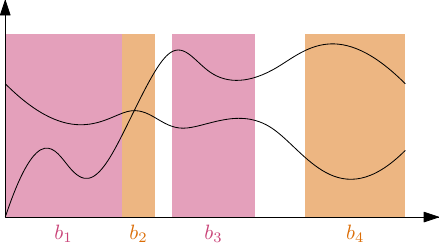}
		\caption{Adversary (Example).}\label{fig:chunkAdv}
	\end{subfigure}
	\caption{
		Discretization of parties.
		Here, honest parties discretize in fixed time intervals, while the adversary may construct blocks in any non-overlapping fashion.
	}\label{Fig:discretechain}
\end{figure}

For honest parties, chain profile and resource profile are identical.
Honest parties discretize their resource profile $\hres$ by following some prescribed rules to create blocks (e.g., in fixed one unit time intervals as depicted in~\cref{fig:chunkHon}).
The resulting blocks are non-overlapping and cover the whole timespan without gaps.
The latter requirement is reasonable since honest parties generally do not waste resources.
Without loss of generality, we assume the time interval to create a block is $1$.

\begin{definition}[Honest Discretization]\label{def:hondisc}
	Let the honest parties' resource profile be $\hres = (\bhsp(t), \bhvdf(t), \bhw(t))_{[0,T]}$.
	Consider a blockchain $\hbchain = (b^\Hc_0, \ldots, b^\Hc_T)$ where each block $b^\Hc_i$ spans the time $(t_i, t'_i)$.
	$\hbchain$ is an honest blockchain arising from $\hres$ if $t_0 = 0$, $t'_T = T$, and $t'_i = i+1$ for all $i \in [T-1]$.
\end{definition}

The adversary also starts from a resource profile $\ares$, but they may cheat when deriving the blockchain from the resources.
In terms of discretization, the adversary may not necessarily follow the prescribed rule.
It may create blocks covering varying timespans, or it might leave gaps between blocks.
The only condition is that blocks don't overlap (as shown in \cref{fig:chunkAdv}).

\begin{definition}[Adversarial Discretization]\label{def:advdisc}
	\sloppy Let the adversary's resource profile $\ares = (\basp(t), \bavdf(t), \baw(t))_{[0,T]}$.
	Consider a blockchain $\abchain = (b^\Ac_0, \ldots, b^\Ac_{B})$ where each block $b^\Ac_i$ spans the time $(t_i, t'_i)$.
	$\abchain$ is an adversarial blockchain arising from $\ares$ if $0 \leq t_0$, $t'_B \leq T$ and $t'_i \leq t_{i+1}$ for all $i \in [B-1]$.
\end{definition}

Looking ahead, the security of the discrete blockchain quantitatively depends on the maximum fluctuation of resources within blocks.
We quantify this fluctuation by the $\xi$-\emph{Smoothness} of resources.
Essentially, $\xi \geq 1$ bounds the absolute change of resources within a block.

\begin{definition}[$\xi$-Smoothness]\label{def:smooth}
	Let $\xi \geq 1$.
	A blockchain $\bchain$ arising from $\res$ satisfies $\xi$-smoothness if, for all blocks $0 \leq i \leq B$, it holds that
	\begin{align*}
		{\bSbmax(b_i)} &\leq \xi\cdot {\bSbmin(b_i)}\\
		{\bVbmax(b_i)} &\leq \xi\cdot {\bVbmin(b_i)}\\
		{\bWbmax(b_i)} &\leq \xi\cdot {\bWbmin(b_i)}.\\
	\end{align*}
\end{definition}

\begin{remark}
        In practice, blockchains generally ensure that resources within a block are relatively smooth, i.e, $\xi$ is small.
        They do so by imposing an upper bound on the resources within a block. 
        For example, Bitcoin's difficulty mechanism essentially puts an upper and lower bound on the work within a block (\cite{KMMNTT21} proposes an alternative way to record work; it has no lower bound, yet still an upper bound).
        This effectively limits the time span a block takes.
        Since physical resources are not very elastic---especially at the quantities that are consumed by blockchains---fluctuation is effectively limited.
\end{remark}

\subsection{Security Statement}
Intuitively, the security statement in the continuous model was:
If the adversary starts out with fewer resources than the honest parties, then the weight of the adversarial chain is lower than that of the honest one.
In the discrete model, the result is a bit weaker because we require a quantitative gap, denoted by $\delta \geq 1$, between honest and adversarial resources.
Together with $\xi$-Smoothness, this leads to the definition of $(\delta, \xi)$-security below.

\begin{definition}[Weight Function Security Against PDS, Discrete Model]\label{def:secdisc}
	A weight function $\Gamma$ is $(\delta, \xi)$-secure against private-double spending in the discrete model if, for all
	honest and adversarial resource profiles $\Rc^\Hc$ and $\Rc^\Ac$
	such that
    \begin{equation}
            \delta \cdot \weight(\basp(t), \bavdf(t), \baw(t)) < \weight(\bhsp(t), \bhvdf(t), \bhw(t)) \: \forall t \in [0, T]\label{eq:discass}
    \end{equation}
	the following holds:
	
	For any $\xi$-smooth (\cref{def:smooth}) blockchains $\hbchain$ and $\abchain$, respectively arising from $\hres$ and $\ares$ according to \cref{def:hondisc,def:advdisc}, it holds that 
	\begin{equation*}
		\bcweight(\hbchain) > \bcweight(\abchain).
	\end{equation*}
\end{definition}

We remark that the adversary is more powerful if $\delta$ is small (i.e., the gap is small) and $\xi$ is large (i.e., resources may fluctuate by a large magnitude).
The following theorem expresses $\xi$  as a function of $\delta$, namely $\xi = \sqrt[4]{\delta}$.
This means that if the gap $\delta$ is small, then only small fluctuations of resources within blocks may be tolerated.

\begin{theorem}[Secure Weight Functions, Discrete Model]\label{thm:discrete}
        For any $\delta \geq 1$, a weight function is $\Gamma(\Sb, \Vb, \Wb)$ is $(\delta, \sqrt[4]{\delta})$-secure against private-double spending (\cref{def:secdisc}) if it is 
        \begin{enumerate}
                \item monotonically increasing;
                \item homogeneous in $\Vb$ and $\Wb$; and
                \item subhomogeneous in $\Sb$.
        \end{enumerate}
\end{theorem}

 We defer the proof to \cref{sec:proofDiscrete}.

\subsection{Replotting Attacks}\label{S:replot}
In the discrete model, we also have to consider \emph{replotting attacks}. Such attacks were first discussed in the Spacemint~\cite{FC:PKFGAP18} paper under the term ``space reuse''. The Chia green paper~\cite{chia2} discusses them in more detail.

Replotting attacks are easiest understood when one assumes that disk space is bound to some public key of a wallet.
Then, in a replotting attack, the adversary repeatedly \emph{replots} (i.e., re-initializes) its space using different keys within the time span of a block.
This effectively increases the adversary's space within a block at the cost of extra computation to perform the replotting.
Concretely, we assume replotting takes $\rho > 1$ time (usually, blockchains ensure that this $\rho$ large), and an adversary with $N$ space that replots $m$ times appears to have $(m+1)\cdot N$ space.

\begin{remark}
        No matter the concrete cryptographic primitive used to track space, such attacks seem unavoidable in the fully-permissionless model.
        In other settings, e.g., the quasi-permissionless model, replotting can be disincentivized.
        For example, Filecoin~\cite{filecoin} requires parties to commit to space, and then the parties must continuously prove that they are storing the committed space.
        This prevents replotting if the gap between the required proofs is smaller than the replotting time.
        In practice, this gap might be too small and hence inefficient, so a more delicate security argument is needed; see~\cite{gn23} for details.
\end{remark}

\subsubsection{Extra Assumptions are Necessary}
Without any extra assumptions, replotting leads to attacks in the discrete model.
For the sake of example, consider Chia's weight function $S \cdot V$, which is secure according to \cref{thm:discrete}.
Assume replotting takes time $\rho=2$, $S^\Ac = V^\Ac = 1$ and $S^\Hc = V^\Hc = 1.1$ (this gap suffices since both profiles are $1$-smooth), and consider the timespan $[0, 6]$.
The honest parties create $5$ blocks with a cumulative weight of $6 \cdot (1.1 \cdot 1.1) \approx 7.3$.
The adversary creates one block in which it replots once.
Assuming that the adversary cannot do anything else while replotting (i.e., it can only gain $V$ for $4$ time), the weight of the block is $4 \cdot (2\cdot 1) = 8$.
Note that this attack generalizes to other weight functions $\weight$ and other values of $\rho$.

\subsubsection{A Solution using Difficulty}
One way to disincentive replotting in the discrete model is bounding the total weight of a block $b$.
Consider that the protocol keeps track of a difficulty $D$ that is periodically adjusted so that roughly one block is created per time unit (e.g., in Bitcoin the difficulty is reset once every two weeks so blocks arrive roughly every 10 minutes). 
Further, let $\eta \geq 1$ be a protocol parameter (to be set later).
Then, the protocol bounds the weight of blocks as $D \leq \weight(b) \leq \eta\cdot D$ (abusing notation of $\weight$ slightly).

If we now set $\eta < \rho$, it is not hard to see that replotting does not help:
Replotting even once requires $\rho$ time, and the resulting adversarial block has at most $\eta\cdot D < \rho \cdot D$ weight.
Meanwhile, the honest parties produce around $\rho$ blocks, each of weight at least $D$; so in total $\rho \cdot D$.

Note that this argument requires $D$ to stay fixed, an attacker might still be able to create a heavier chain over a long period of time that spans several epochs (where the difficulty is reset once every epoch).

Chia~\cite{chia2} with its weight function $\Gamma(S,V)=S\cdot V$ uses a similar idea, but it tracks the difficulty of the space and VDF separately. 
The block arrival time is only determined by the VDF difficulty, which is nice as VDF speed hardly fluctuates at all over time.

One can generalize this idea to any weight function that can be written as $\Gamma(\Sb, \Vb, \Wb) = \Gamma_1(\Sb) \cdot \Gamma_2(\Vb, \Wb)$. 
Now one would require that each block that records resources $\vb,\wb,\ssb$ must satisfy $\Gamma_2(\vb,\wb) = D$. One doesn't need to put an additional upper bound on the space 
$\Gamma_1(\ssb)$ if $\Gamma_1$ is subhomogenous, i.e., for any $\alpha>1$ we have $\Gamma_1(\alpha\ssb)\le \alpha\Gamma_1(\ssb)$ (Chia does both, it is (sub)homogenous and has an upper bound).

\subsubsection{Future Work on Replotting}
As mentioned, the above solutions don't formally prevent replotting attacks that range over many epochs. In practice that might not be such a big issue, as extremely long range attacks are not really practical: they  require a lot of resources for a long period of time, and thus are expensive to launch.  Moreover it might be difficult to convince honest parties to accept a very long fork as it's such an obvious attack. It still would be interesting to understand whether it's possible to formally achieve security against replotting attacks, we leave this for future work.

 \paragraph{Financial Conflicts-of-Interest} %
Krzysztof Pietrzak is a scientific advisor to Chia Network Inc.

 \paragraph{Acknowledgements}
 This research was funded in whole or in part by the Austrian Science Fund (FWF) 10.55776/F85. For open access purposes, the author has applied a CC BY public copyright license to any author-accepted manuscript version arising from this submission.

\bibliography{cryptobib/abbrev3,cryptobib/crypto,references}

 \appendix
 \section{Proof of \cref{thm:idealModSec}}\label{sec:idealProof}
\begin{theorem}[Secure Weight Functions, Continuous Model; \cref{thm:idealModSec} restated]
        A weight function $\weight$ is \emph{secure against private double-spending attacks} in the \emph{continuous model} if and only if $\weight(\Sb,\Vb,\Wb)$ is monotonically increasing (\cref{def:monotone}) and homogeneous in $\Vb,\Wb$ (\cref{def:homogenous}).
\end{theorem}

We will prove the theorem in three parts using \cref{lem:secIfMonoAndHomo,lem:notSecIfNotMono,lem:notSecIfNotHomo}.

\begin{lemma}[If-Direction of \cref{thm:idealModSec}]\label{lem:secIfMonoAndHomo}
        If $\weight(\Sb,\Vb,\Wb)$ is monotonically increasing and $\weight(\Sb,\Vb,\Wb)$ is homogeneous in $\Vb,\Wb$, then $\cweight(\ares) \geq \cweight(\achain)$ where $\achain$ satisfies \cref{def:idealAttack} for $\ares$ and any $\phi(t)$.

        As a consequence, $\weight$ is secure against PDS if $\weight(\Sb,\Vb,\Wb)$ is monotonically increasing and $\,\weight(\Sb,\Vb,\Wb)$ is homogeneous in $\Vb,\Wb$.
\end{lemma}
\begin{proof}
 the first part of the lemma, consider the adversarial chain profile $\achain$ from~\cref{def:idealAttack}. For any $\ttil \in [0, \Ttilend]$, $\advy$ could create a chain profile such that 
\begin{align*}
	0 &< \basptil(\ttil) \le \basp(\atinv(\ttil)),\\
	0 &< \bavdftil(\ttil) \le \phi(\atinv(\ttil)) \cdot \bavdf(\atinv(\ttil)),\\
	0 &< \bawtil(\ttil) \le \phi(\atinv(\ttil)) \cdot \baw(\atinv(\ttil)).
\end{align*}
Since $\weight$ is monotonic,
\begin{align*}
	&\weight(\basptil(\ttil), \bavdftil(\ttil), \bawtil(\ttil) ) \le \\ &\quad \weight(\basp(\atinv(\ttil)), \phi(\atinv(\ttil)) \cdot \bavdf(\atinv(\ttil)), \phi(\atinv(\ttil)) \cdot \baw(\atinv(\ttil)))
\end{align*}
holds for all $\ttil \in [0, \Ttilend]$. 
Since $\weight$ is also homogeneous in $\Vb,\Wb$, 
\begin{align*}
	&\weight(\basptil(\ttil), \bavdftil(\ttil), \bawtil(\ttil) ) \le \\ &\quad \phi(\atinv(\ttil)) \cdot \weight(\basp(\atinv(\ttil)),  \bavdf(\atinv(\ttil)), \baw(\atinv(\ttil))),
\end{align*}
so we can conclude that
\begin{align*}
	\cweight(\achain) &= \int_{0}^{\Ttilend} \weight(\basptil(\ttil), \bavdftil(\ttil), \bawtil(\ttil)) \,d\ttil\\	
	&\le \int_{0}^{\Ttilend} \phi(\atinv(\ttil)) \cdot \weight(\basp(\atinv(\ttil)), \bavdf(\atinv(\ttil)), \baw(\atinv(\ttil))) \, d\ttil.
\end{align*}
Now, we integrate by substituting\footnote{$\int_{a}^{b}f(g(x))g'(x) \,dx = \int_{g(a)}^{g(b)}f(u) \,du$} $t = \atinv(\ttil)$.
Here, note that $\frac{d}{d\ttil}\atinv(\ttil) = \phi(\atinv(\ttil))$ by the inverse function rule.\footnote{$\frac{d}{dx}f^{-1}(a) = \big(f^{-1}\big)'(a) = \frac{1}{f'(f^{-1}(a))}$}
This leads to
\begin{align*}
	\cweight(\achain) &\leq \int_{\atinv(0)}^{\atinv(\Ttilend)} \phi(T) \cdot \weight(\basp(t), \bavdf(t), \baw(t)) \cdot \frac{1}{
		\phi(t)} \, dt &\\
	&= \int_{0}^{T_\eend} \phi(t) \cdot \weight(\basp(t), \bavdf(t), \baw(t)) \cdot \frac{1}{\phi(t)} \, dt &\\
	&= \int_{0}^{T_\eend} \weight(\basp(t), \bavdf(t), \baw(t)) \, dt &\\
	&= \cweight(\ares).
\end{align*}
This proves the first part of the lemma.

For the second part, note that the preconditions on resources in \cref{def:secure} imply that
\begin{equation*}
	\cweight(\hres) > \cweight(\ares).
\end{equation*}
By the first part of this lemma and since $\hres = \hchain$ by \cref{def:secure}, the second part follows. This completes the proof.
\end{proof}

\begin{lemma}[Only-If-Direction of \cref{thm:idealModSec}, Part I]\label{lem:notSecIfNotMono}
	$\weight$ is not secure against PDS if $\weight(\Sb, \Vb, \Wb)$ is not monotonically increasing. 
\end{lemma}
\begin{proof}
Suppose $\weight$ is not monotonically increasing, i.e., there exist $(\ssb, \vb, \wb)$ and $(\ssb', \vb', \wb')$ such that $(\ssb, \vb, \wb) < (\ssb', \vb', \wb')$ but $\weight(\ssb, \vb, \wb) > \weight(\ssb', \vb', \wb')$. 
In this case, the adversary can simply put less resources in the adversarial chain than it actually has to get a chain profile of higher weight. 

Formally, for some time $T_\eend>0$, consider the resource profiles
\begin{align*}
	\bhsp(t) &= \ssb, &\: \bhvdf(t) &= \vb, &\: \bhw(t) &= \wb &\quad &\text{for } t \in [0,T_\eend] \\
	\basp(t) &= \ssb'
	, &\: \bavdf(t) &= \vb', &\: \baw(t) &= \wb' &\quad &\text{for } t \in [0,T_\eend].
\end{align*} 
Clearly, the weight of adversarial resources is strictly less than honest resources at every point of time. Now for adversarial chain (\cref{def:idealAttack}) $\advy$ chooses $\phi(t) = 1$ for $t\in[0,T]$. 
Thus, $\at(t) = \atinv(t) = t$ and $\Ttil_\eend = T_\eend$. Then $\advy$ choose 
\begin{align*}
	\basptil(\ttil) &= \ssb \leq \phi(T) \cdot \basp(T) = \ssb'\\ 
	\bavdftil(\ttil) &= \vb \leq \phi(T) \cdot \bavdf(T) = \vb'\\
	\bawtil(\ttil) &=  \wb \leq \phi(T) \cdot \baw(T) = \wb'
\end{align*} for all $\ttil \in [0, \Ttil_\eend]$, where $T = \atinv(\ttil)$. 

Thus,
\begin{align*}
	\cweight(\achain) &= \int_{0}^{\Ttil_\eend} \weight(\ssb, \vb, \wb)\\ 
	&= \int_{0}^{T_\eend} \weight(\ssb, \vb, \wb) = \cweight(\hchain).
\end{align*}
Therefore, $\weight$ is not secure. This completes the proof.
\end{proof}

\begin{lemma}[Only-If-Direction of \cref{thm:idealModSec}, Part II]\label{lem:notSecIfNotHomo}
        $\weight$ is not secure against PDS if $\,\weight(\Sb,\Vb,\Wb)$ is not homogeneous in $\Vb,\Wb$.
\end{lemma}
\begin{proof}
Due to \cref{def:secure}, if $\weight$ is constant then it is not secure as the preconditions on the resource profiles can not be met. In that case, we are done. From hereon we assume $\weight$ is not a constant function.

Due to~\cref{lem:notSecIfNotMono} we can assume that $\weight(\Sb,\Vb,\Wb)$ is monotonically increasing in $\Sb,\Vb,\Wb$. Suppose $\weight(\Sb,\Vb,\Wb)$ is not homogeneous in $\Vb,\Wb$, i.e., there exists $\alpha > 0$ and $(\ssb,\vb,\wb) \in \RRgz^{k_1+k_2+k_3}$ such that $\weight(\ssb,  \alpha \vb, \alpha \wb) \neq \alpha \weight(\ssb, \vb, \wb).$ Now we have two cases: 
\begin{itemize}
	\item\textbf{Case 1:} $\weight(\ssb, \alpha \vb, \alpha \cdot \wb) > \alpha \weight(\ssb,\vb,\wb)$. 
	
	\item\textbf{Case 2:} $\weight(\ssb, \alpha \cdot \vb, \alpha \cdot \wb) < \alpha \weight(\ssb,\vb,\wb)$
	
	This implies $ \weight(\ssb, \frac{1}{\alpha} \cdot \vb', \frac{1}{\alpha} \wb') > \frac{1}{\alpha} \cdot \weight(\ssb, \vb', \wb')$ where $\vb' = \alpha \vb, \wb' = \alpha \cdot \wb$.
	Since $\frac{1}{\alpha} > 0$, this case reduces to Case 1. 
\end{itemize}
For \textbf{Case 1}, $\weight(\ssb, \alpha \vb, \alpha \wb) > \alpha \weight(\ssb,\vb,\wb)$ is equivalent to 
\begin{equation}\label{eq:nonHomoeq1}
	\weight(\ssb, \alpha \vb, \alpha \wb) = \alpha \weight(\ssb,\vb,\wb) + \beta  
\end{equation} for some $\beta \in \RRgz$. 
Note that $\alpha = 1$ implies $	\weight(\ssb, \alpha \vb, \alpha \cdot \wb) =\weight(\ssb,\vb,\wb) =  \alpha \weight(\ssb,\vb,\wb) + \beta = \weight(\ssb,\vb,\wb) + \beta$. 
Which in turn implies $\beta = 0$, a contradiction. 
Thus, $\alpha \neq 1$.

\textbf{Case 1} can be further divided in sub-cases: 
\begin{itemize}
	\item \textbf{Case 1a:}  $\alpha > 1$ and $\weight(\ssb , \alpha \vb, \alpha \wb)  \leq \weight(\ssb, \vb, \wb)$.
	\item \textbf{Case 1b:} $\alpha > 1$ and $\weight(\ssb , \alpha \vb, \alpha \wb) > \weight(\ssb, \vb, \wb)$
	\item \textbf{Case 1c:} $\alpha < 1$ and $\weight(\ssb , \alpha \vb, \alpha \wb) < \weight(\ssb, \vb, \wb)$.
	\item \textbf{Case 1d:} $\alpha < 1$ and $\weight(\ssb , \alpha \vb, \alpha \wb) \geq \weight(\ssb, \vb, \wb)$.
\end{itemize}
Let's prove each case individually:
\begin{itemize}
	\item[] \textbf{Case 1a:} $\alpha > 1$ and $\weight(\ssb , \alpha \vb, \alpha \wb)  \leq \weight(\ssb, \vb, \wb)$. Since $(\ssb, \vb, \wb) < (\ssb, \alpha \vb, \alpha, \wb)$ and due to monotonicity of $\weight$ (\cref{lem:notSecIfNotMono}), we also have that $\weight(\ssb, \alpha \vb, \alpha \wb) \geq \weight(\ssb, \vb, \wb)$. Thus, $\weight(\ssb, \alpha \vb, \alpha \wb) = \weight(\ssb, \vb, \wb)$. Using~\cref{eq:nonHomoeq1}, we get $$\weight(\ssb, \alpha \vb, \alpha \wb) = \alpha \weight(\ssb, \vb, \wb) + \beta = \weight(\ssb, \vb, \wb). $$ This implies, $(1-\alpha)\weight(\ssb, \vb, \wb) = \beta$. Since $\alpha>1$, the left-hand side is negative while right-hand side is positive. Hence, this case is impossible.\\
	
	\item[] \textbf{Case 1b:} $\alpha > 1$ and $\weight(\ssb , \alpha \vb, \alpha \wb) > \weight(\ssb, \vb, \wb)$. In this case \enquote{squeezing} time gives more weight than the original resources profile. 
	$\advy$ will \enquote{squeeze} $(\vb, \wb)$ by factor $\alpha $ to reach $(\alpha \vb, \alpha \wb)$ and use this to get a higher weight than the honest chain profile.
	
	Formally, for $T_\eend = T_0 + T_1$ where $T_1 = 1$ and $T_0 \geq \frac{\alpha - 1}{\beta} \cdot \weight(\ssb, \alpha \vb, \alpha \wb)$ consider resource profiles $\Rc^{\Hc}$ and $\Rc^{\Ac}$ such that:
	\begin{align*}
		\bhsp(t) &= \ssb, &\: \bhvdf(t) &= \vb, &\: \bhw(t) &= \wb &\quad &\text{for } t \in [0,T_0) \\
		\bhsp(t) &= \ssb, &\: \bhvdf(t) &= \alpha \vb, &\: \bhw(t) &= \alpha \wb &\quad &\text{for } t \in [T_0,T_\eend] \\
		\basp(t) &= \ssb, &\: \bavdf(t) &= \vb, &\: \baw(t) &= \wb &\quad &\text{for } t \in [0,T_\eend]
	\end{align*}
	Since $\weight(\bhsp(t), \bhvdf(t), \bhw(t)) \geq \weight(\basp(t), \bavdf(t), \baw(t))$ for all $t\in[0,T_\eend]$ and $\weight(\bhsp(t), \bhvdf(t), \bhw(t)) > \weight(\basp(t), \bavdf(t), \baw(t))$ for all $t\in [T_0,T_\eend]$, preconditions on resource profiles of \cref{def:secure} are satisfied.
	
	The weight of the honest chain profile is 
	\begin{align*}
		\cweight(\hchain) &= T_0 \cdot \weight(\ssb, \vb, \wb) + T_1 \cdot \weight(\ssb, \alpha \vb, \alpha \wb) &\\
		&=  T_0 \cdot \weight(\ssb, \vb, \wb) + T_1 \cdot \alpha \cdot \weight(\ssb, \vb, \wb) + T_1 \cdot \beta &(\text{by}~\cref{eq:nonHomoeq1})
	\end{align*} 
	$\advy$ chooses $\phi(t) = \alpha$ for all $t\in[0, T_\eend]$. 
	This gives $\at(T) = \frac{T}{\alpha}$, $\atinv(\ttil) = \alpha \ttil$ and $\Ttil_\eend = \frac{T_\eend}{\alpha}$. Setting $\phi(t) = \alpha$ is \enquote{squeezing} as $\alpha > 1$. 
	$\advy$ chooses
	\begin{align*}
		\basptil(\ttil) &= \basp(T) = \ssb\\
		\bavdftil(\ttil) &= \phi(T) \cdot \bavdf(T) = \alpha \vb\\
		\bawtil(\ttil) &= \phi(T) \cdot \baw(T) = \alpha \wb
	\end{align*}
	for all $\ttil \in [0, \Ttil_\eend]$, where $T = \atinv(\ttil) = \alpha \ttil$.
	
	Thus, the weight of the adversarial chain profile is 
	\begin{align*}
		\cweight(\achain) &= \int_{0}^{\Ttil_\eend} \weight(\ssb, \alpha \vb, \alpha \wb) \,dt =  \Ttil_\eend \cdot \weight(\ssb, \alpha \vb, \alpha \wb) &&\\
		&= \frac{T_\eend}{\alpha} \cdot \weight(\ssb, \alpha \vb, \alpha \wb) = \frac{T_0+T_1}{\alpha} \cdot \weight(\ssb, \alpha \vb, \alpha \wb)&\\
		&= \frac{T_0+T_1}{\alpha} \cdot (\alpha \weight(\ssb, \vb, \wb) + \beta)& &\text{by}~\cref{eq:nonHomoeq1}&\\
		&\text{Since } T_0 \geq \frac{\alpha - 1}{\beta} \cdot \weight(\ssb, \alpha\vb, \alpha \wb ) \text{ and } T_1 = 1, &\\&\text{by simplifying, we get}&\\
		&\geq T_0 \weight(\ssb, \vb, \wb)  + T_1(\alpha \weight(\ssb, \vb, \wb) + \beta)\\
		&= \cweight(\hchain).&  
	\end{align*}
	This implies $\cweight(\achain) \geq \cweight(\hchain)$ and hence $\weight$ is not secure.\\

        \item[] \textbf{Case 1c:} $\alpha < 1$ and $\weight(\ssb , \alpha \vb, \alpha \wb) < \weight(\ssb, \vb, \wb)$. This case is the reverse of the previous case. Here \enquote{stretching} leads to a higher weight than the original resource profile. $\advy$ \enquote{stretches} $(\vb, \wb)$ by a factor $\alpha$ into $(\alpha \vb, \alpha \wb)$ in order to get a higher weighted chain profile than the honest chain profile. 

                Formally, let $T_\eend = T_0 + T_1$ where $T_1 = 1$ and $T_0 \geq \frac{\alpha}{\beta}((1-\alpha) \weight(\ssb, \vb, \wb) - \beta), \: T_0 > 0$ and consider the resource profiles $\Rc^{\Hc}$ and $\Rc^{\Ac}$: 
                \begin{align*}
                        \bhsp(t) &= \ssb, &\: \bhvdf(t) &= \vb, &\: \bhw(t) &= \wb &\quad &\text{for } t \in [0,T_\eend] \\
                        \basp(t) &= \ssb, &\: \bavdf(t) &= \vb, &\: \baw(t) &= \wb &\quad &\text{for } t \in [0,T_0) \\
                        \basp(t) &= \ssb, &\: \bavdf(t) &= \alpha \vb, &\: \baw(t) &= \alpha \wb &\quad &\text{for } t \in [T_0,T_\eend] \\
                \end{align*} \sloppy Since $\weight(\bhsp(t), \bhvdf(t), \bhw(t)) \geq \weight(\basp(t), \bavdf(t), \baw(t))$ for all $t\in[0,T_\eend]$ and $\weight(\bhsp(t), \bhvdf(t), \bhw(t)) > \weight(\basp(t), \bavdf(t), \baw(t))$ for all $t\in [T_0,T_\eend]$, the preconditions on resource profiles in \cref{def:secure} are met.

                The weight of the honest chain profile is 
                \begin{align*}
                        \cweight(\hchain) &= T_\eend \cdot \weight(\ssb, \vb, \wb)=(T_0 + T_1) \cdot \weight(\ssb, \vb, \wb). 
                \end{align*} 

                $\advy$ sets $\phi(t) = \alpha$ for all $t \in [0, T_0)$ and $\phi(t) = 1$ for all $t \in [T_0, T_1]$, which gives 
                \[
                        \at(T) = \begin{cases}
                                \frac{T}{\alpha} &\text{for all} \: t \in [0,T_0)\\
                                \frac{T_0}{\alpha} + (T-T_0) &\text{for all} \: t \in [T_0, T_1]
                        \end{cases}
                \]

                \[
                        \atinv(\ttil) = \begin{cases}
                                \alpha \ttil &\text{for all} \: \ttil \in [0,\frac{T_0}{\alpha})\\
                                T_0 + (\ttil - \frac{T_0}{\alpha}) &\text{for all} \: \ttil \in [\frac{T_0}{\alpha}, \Ttil_\eend]
                        \end{cases}
                \] 

                and $\Ttil_\eend = \frac{T_0}{\alpha} + T_1$. Setting $\phi(t) = \alpha$ is \enquote{stretching} as $\alpha < 1$.\\

                Now $\advy$ chooses 
                \begin{align*}
                        \basptil(\ttil) &= \basp(T)\\
                        \bavdftil(\ttil) &= \phi(T) \cdot \bavdf(T)\\
                        \bawtil(\ttil) &= \phi(T) \cdot \baw(T)
                \end{align*} for all $\ttil \in [0, \Ttil_\eend]$, where $T = \atinv(\ttil) = \alpha \ttil$. \\

                Thus, the weight of the adversarial chain profile is 
                \begin{align*}
                        \cweight(\achain) &= \int_{0}^{\at(T_0)} \weight(\basptil(t), \bavdftil(t), \bawtil(t))\,dt \\
                                          &\qquad + \int_{\at(T_0)}^{\Ttil_\eend} \weight(\basptil(t), \bavdftil(t), \bawtil(t))\,dt&\\
                                          &= \int_{0}^{\frac{T_0}{\alpha}} \weight(\ssb, \alpha \vb, \alpha \wb) \,dt\\ &\qquad + \int_{\frac{T_0}{\alpha}}^{\frac{T_0}{\alpha} + T_1} \weight(\ssb, \alpha \vb, \alpha \wb) \,dt\\
                                          &= \frac{T_0}{\alpha} \cdot \weight(\ssb, \alpha \vb, \alpha \wb) + T_1 \cdot \weight(\ssb, \alpha \vb, \alpha \wb)\\
                                          &= \left(\frac{T_0}{\alpha} + T_1 \right)(\alpha \weight(\ssb, \vb, \wb) + \beta) &\text{by} \:~\cref{eq:nonHomoeq1} \\
                                          &\text{Since } T_0 \geq \frac{\alpha}{\beta}((1-\alpha)\weight(\ssb, \vb, \wb) - \beta) \text{ and } T_1 = 1, &\\&\text{by simplifying, we get}&\\
                                          &\geq (T_0 + T_1) \cdot \weight(\ssb, \vb, \wb) \\
                                          &= \cweight(\hchain).
                \end{align*}
                This implies $\cweight(\achain) \geq \cweight(\hchain)$ and thus $\weight$ is not secure.\\

        \item[] \textbf{Case 1d:} $\alpha < 1$ and $\weight(\ssb , \alpha \vb, \alpha \wb) \geq \weight(\ssb, \vb, \wb)$. Since $(\ssb, \alpha \vb, \alpha \wb) < (\ssb, \vb, \wb)$, by monotonicity~\cref{lem:notSecIfNotMono} we have that $\weight(\ssb, \alpha \vb, \alpha \wb) \leq \weight(\ssb, \vb, \wb)$. Thus, $\weight(\ssb, \alpha \vb, \alpha \wb) = \weight(\ssb, \vb, \wb)$. Intuitively, this says that stretching by factor $\frac{1}{\alpha}$ doesn't change the weight but since it increases the time as well it will give higher weight to the resulting chain profile.\\

                To show this formally we need to find two points in resource space such that weight varies among the two points. Since $\weight$ is not constant, there exists $(\ssb', \vb', \wb') \in \RRgz^3$ such that $\weight(\ssb', \vb', \wb') \neq \weight(\ssb, \vb, \wb) = \weight(\ssb, \alpha \vb, \alpha \wb)$. 

                Let $\delta \eqq |\weight(\ssb', \vb', \wb') - \weight(\ssb, \vb, \wb)|.$ \\

                We have two cases: \begin{itemize}
                        \item[] \textbf{Case A:} $\weight(\ssb', \vb', \wb') > \weight(\ssb, \vb, \wb)$ 
                        \item[] \textbf{Case B:} $\weight(\ssb', \vb', \wb') < \weight(\ssb, \vb, \wb)$. 
                \end{itemize}
                We describe the violation of \cref{def:secure} in both cases together while highlighting the differences in the steps as we go: 
                Let $T_\eend = T_0 + T_1$ where $T_1 = 1$ and $T_0 \geq \frac{\delta}{\weight(\ssb, \vb, \wb)\cdot(\frac{1}{\alpha} - 1)}.$

                Consider the resource profiles $\Rc^{\Hc} = (\bhsp(t), \bhvdf(t), \bhw(t))$ and $\Rc^{\Ac} = (\basp(t), \bavdf(t), \baw(t))$ such that: \begin{align*}
                        \bhsp(t) &= \ssb, &\: \bhvdf(t) &= \vb, &\: \bhw(t) &= \wb &\quad &\text{for } t \in [0,T_0) \\
                        \basp(t) &= \ssb, &\: \bavdf(t) &= \vb, &\: \baw(t) &= \wb &\quad &\text{for } t \in [0,T_0] \\
                        \textbf{Case A:}:\\
                        \bhsp(t) &= \ssb', &\: \bhvdf(t) &= \vb', &\: \bhw(t) &= \wb' &\quad &\text{for } t \in [T_0,T_\eend] \\
                        \basp(t) &= \ssb, &\: \bavdf(t) &= \vb, &\: \baw(t) &= \wb &\quad &\text{for } t \in [T_0,T_\eend]\\
                        \textbf{Case B:}:\\
                        \bhsp(t) &= \ssb, &\: \bhvdf(t) &= \vb, &\: \bhw(t) &= \wb &\quad &\text{for } t \in [T_0,T_\eend] \\
                        \basp(t) &= \ssb', &\: \bavdf(t) &= \vb', &\: \baw(t) &= \wb' &\quad &\text{for } t \in [T_0,T_\eend]
                \end{align*} Note that in both cases we have an interval where $\advy$'s resources has strictly lower weight than the $\Hc$'s resources. Thus, it satisfies the precondition on resource profiles in \cref{def:secure}.  

                The weight of the honest chain profile is: 
                \begin{equation*}
                        \cweight(\hchain) = \begin{cases}
                                T_0 \cdot \weight(\ssb, \vb, \wb) + T_1 \cdot \weight(\ssb', \vb', \wb') &\text{for }\textbf{Case A}\\
                                T_0 \cdot \weight(\ssb, \vb, \wb) + T_1 \cdot \weight(\ssb, \vb, \wb) &\text{for }\textbf{Case B} 
                        \end{cases}
                \end{equation*} which, by definition of $\delta$, is same as:
                \begin{equation*}
                        \cweight(\hchain) = \begin{cases}
                                T_0 \cdot \weight(\ssb, \vb, \wb) + \weight(\ssb, \vb, \wb) + \delta &\text{for }\textbf{Case A}\\
                                T_0 \cdot \weight(\ssb, \vb, \wb) + \weight(\ssb', \vb', \wb') + \delta &\text{for }\textbf{Case B} 
                        \end{cases}
                \end{equation*}

                $\advy$ chooses $\phi(t) = \alpha$ for all $t \in [0, T_0)$ and $\phi(t) = 1$ for all $t \in [T_0, T_1]$. This intuitively gives us a stretch by factor $\frac{1}{\alpha}$ (as $\alpha<1$) for $[0,T_0]$ and the remaining time remains the same.

                We get
                \[
                        \at(T) = \begin{cases}
                                \frac{T}{\alpha} &\text{for all} \: t \in [0,T_0)\\
                                \frac{T_0}{\alpha} + (T-T_0) &\text{for all} \: t \in [T_0, T_1]
                        \end{cases}
                \]

                \[
                        \atinv(\ttil) = \begin{cases}
                                \alpha \ttil &\text{for all} \: \ttil \in [0,\frac{T_0}{\alpha})\\
                                T_0 + (\ttil - \frac{T_0}{\alpha}) &\text{for all} \: \ttil \in [\frac{T_0}{\alpha}, \Ttil_\eend]
                        \end{cases}
                \] 

                and $\Ttil_\eend = \frac{T_0}{\alpha} + T_1$.

                $\advy$ chooses 
                \begin{align*}
                        \basptil(\ttil) &= \basp(T) = \ssb\\ 
                        \bavdftil(\ttil) &= \phi(T) \cdot \bavdf(T) = \alpha \vb\\
                        \bawtil(\ttil) &= \phi(T) \cdot \baw(T) = \alpha \wb
                \end{align*} for all $\ttil \in [0, \at(T_0)]$ and
                \begin{align*}	
                        \basptil(\ttil) &= \phi(T) \cdot \basp(T)\\
                        \bavdftil(\ttil) &= \phi(T) \cdot \bavdf(T)\\
                        \bawtil(\ttil) &= \phi(T) \cdot \baw(T)
                \end{align*} for all $\ttil \in [\at(T_0), \Ttil_\eend]$ where $T = \atinv(\ttil) = \alpha \ttil$. \\

                Thus, the weight of adversarial chain profile is
                \begin{align*}
                        \cweight(\achain) &= \int_{0}^{\Ttil_\eend} \weight(\basptil(t), \bavdftil(t), \bawtil(t)) \, dt &\\
                                          &= \int_{0}^{\at(T_0)} \weight(\ssb, \alpha \vb, \alpha \wb) \,dt \\
                                          &+ \int_{\at(T_0)}^{\Ttil_\eend} \weight(\basptil(t), \bavdftil(t), \bawtil(t)) \, dt&
                \end{align*}
                For \textbf{Case A:}, 
                \begin{align*}
                        \cweight(\achain) &= \int_{0}^{\frac{T_0}{\alpha}} \weight(\ssb, \vb, \wb) \, dt + \int_{\frac{T_0}{\alpha}}^{\frac{T_0}{\alpha} + T_1} \weight(\ssb,\vb,\wb)\,dt\\
                                          &= \frac{T_0}{\alpha} \cdot \weight(\ssb, \vb, \wb) + T_1 \cdot \weight(\ssb, \vb, \wb) \\
                                          &\text{Since } T_0 \geq \frac{\delta}{\weight(\ssb, \vb, \wb) \cdot (\frac{1}{\alpha} - 1)} \text{ and } T_1 = 1, &\\&\text{plugging in and simplifying, we get}&\\
                                          &\geq \cweight(\hchain)
                \end{align*}
                For \textbf{Case B:}, 
                \begin{align*}
                        \cweight(\achain) &= \int_{0}^{\frac{T_0}{\alpha}} \weight(\ssb, \vb, \wb) \, dt + \int_{\frac{T_0}{\alpha}}^{\frac{T_0}{\alpha} + T_1} \weight(\ssb',\vb',\wb')\,dt\\
                                          &= \frac{T_0}{\alpha} \cdot \weight(\ssb, \vb, \wb) + T_1 \cdot \weight(\ssb', \vb', \wb') \\
                                          &\text{Since } T_0 \geq \frac{\delta}{\weight(\ssb, \vb, \wb) \cdot (\frac{1}{\alpha} - 1)} \text{ and } T_1 = 1, &\\&\text{plugging in and simplifying, we get}&\\
                                          &\geq \cweight(\hchain)
                \end{align*}
                Thus, in either case we get $\cweight(\achain) \geq \cweight(\hchain)$, and hence $\weight$ is not secure. 
\end{itemize}

This completes the proof. 
\end{proof}

\section{Proof of \cref{thm:discrete}}\label{sec:proofDiscrete}
\begin{theorem}[Secure Weight Functions, Discrete Model; \cref{thm:discrete} restated]
        For any $\delta \geq 1$, a weight function is $\Gamma(\Sb, \Vb, \Wb)$ is $(\delta, \sqrt[4]{\delta})$-secure against private-double spending (\cref{def:secdisc}) if it is 
        \begin{enumerate}
                \item monotonically increasing;
                \item homogeneous in $\Vb$ and $\Wb$; and
                \item subhomogeneous in $\Sb$.
        \end{enumerate}
\end{theorem}
\begin{proof}
Consider resource profiles $\hres = (\bhsp(t), \bhvdf(t), \bhw(t))_{[0,T]}$ and $\ares = (\basp(t), \allowbreak\bavdf(t), \allowbreak\baw(t))_{[0,T]}$
with $\cweight(\hres) > \delta\cdot\cweight(\ares)$.
Let $\xi = \sqrt[4]{\delta}$ and consider the blockchains $\hbchain = (b^\Hc_0, \ldots, b^\Hc_{T-1})$ and $\abchain = (b^\Ac_0, \ldots, b^\Ac_{B-1})$
which arise from the resource profiles and are $\xi$-smooth.

We will now prove the sequence of inequalities 
\begin{equation*}
	\bcweight(\hbchain) \geq \frac{1}{\xi^2} \cweight(\hres)  > \xi^2\cdot\cweight(\ares) \geq \bcweight(\abchain),
\end{equation*}
which implies the theorem since $\cweight(\ares) \cdot \xi^4 < \cweight(\hres)$ due to \cref{eq:discass}.
We will prove the left and right inequality separately, using one and two lemmas, respectively.

\textbf{Case $\bcweight(\hbchain) \geq \frac{1}{\xi^2} \cweight(\hres)$:} 
By definition of every block $b_i$ with timespan $(t_i, t'_i)$, it follows that
\begin{align*}
	\bcweight(\hbchain) &= \sum_{b_i \in \hbchain} \weight\left(\bSb^\Hc(b_i), \bVb^\Hc(b_i), \bWb^\Hc(b_i)\right) \\
	&= \sum_{b_i \in \hbchain} \weight\left(\bSb^\Hc(b_i), \int_{t_i}^{t'_i} \Vb^\Hc(t) \, dt, \int_{t_i}^{t'_i} \Wb^\Hc(t) \, dt\right)\\
	&\geq \sum_{b_i \in \hbchain} \weight\left(\bSbmin^\Hc(b_i), \int_{t_i}^{t'_i} \Vb^\Hc(t) \, dt, \int_{t_i}^{t'_i} \Wb^\Hc(t)\, dt\right).
\end{align*}
The third line follows by the monotonicity of $\weight$ and the fact that $\bSb^\Hc(b_i) \geq \bSbmin^\Hc(b_i)$ necessarily. 

Let $\overline{\bVb}^\Hc(b_i) = \displaystyle\frac{1}{t'_i-t_i} \cdot \int_{t_i}^{t'_i} \Vb^\Hc(t) \, dt$ denote the average VDF speed within a block.
Clearly, $\bVb\inf^\Hc(b_i) \leq \overline{\bVb}^\Hc(b_i) \leq \bVbmax^\Hc(b_i)$.
Define $\overline{\bWb}^\Hc(b_i)$ analogously.
Using these insights, we continue with
\begin{align*}
	\weight(\hbchain) &\geq \sum_{b_i \in \hbchain} \weight\left(\bSbmin^\Hc(b_i), \int_{t_i}^{t'_i} \Vb^\Hc(t)\, dt, \int_{t_i}^{t'_i} \Wb^\Hc(t)\, dt\right)\\
	&= \sum_{b_i \in \hbchain} (t'_i-t_i) \cdot \weight\left(\bSbmin^\Hc(b_i), \overline{\bVb}^\Hc(b_i), \overline{\bWb}^\Hc(b_i)\right)\\
	&\geq \sum_{b_i \in \hbchain} (t'_i-t_i) \cdot \weight\left(\bSbmin^\Hc(b_i), \bVbmin^\Hc(b_i), \bWbmin^\Hc(b_i)\right)
\end{align*}
where the second line follows as $\weight$ is homogeneous in $\Vb, \Wb$ and the last line follows from monotonicity.

Now we invoke \cref{def:smooth} to switch $\inf$ to $\sup$, that is,
\begin{align*}
	\bcweight(\hbchain) &\geq \sum_{b_i \in \hbchain} (t'_i-t_i) \cdot \weight(\bSbmin^\Hc(b_i), \bVbmin^\Hc(b_i), \bWbmin^\Hc(b_i))\\
	&\geq \frac{1}{\xi} \sum_{b_i \in \hbchain} (t'_i-t_i) \cdot \weight(\bSbmax^\Hc(b_i), \bVbmin^\Hc(b_i), \bWbmin^\Hc(b_i))\\
	&= \frac{1}{\xi^2} \sum_{b_i \in \hbchain} (t'_i-t_i) \cdot \weight(\bSbmax^\Hc(b_i), \bVbmax^\Hc(b_i), \bWbmax^\Hc(b_i))
\end{align*}
where the second line follows from $\weight$ being sub-homogeneous in $\Sb$ and the third from the homogeneity of $\weight$ in $(\Vb, \Wb)$.

This implies the desired inequality because
\begin{align*}
	\bcweight(\hbchain) &\geq \frac{1}{\xi^2} \sum_{b_i \in \hbchain} (t'_i-t_i) \cdot \weight(\bSbmax^\Hc(b_i), \bVbmax^\Hc(b_i), \bWbmax^\Hc(b_i))\\
	&\geq \frac{1}{\xi^2} \sum_{b_i \in \hbchain} \int_{t_i}^{t'_i}\weight(\Sb^\Hc(t), \Vb^\Hc(t), \Wb^\Hc(t)) \, dt\\
	&= \frac{1}{\xi^2} \int_0^T \weight(\Sb^\Hc(t), \Vb^\Hc(t), \Wb^\Hc(t)) \, dt\\
	&= \frac{1}{\xi^2} \cweight(\hres).
\end{align*}
Note that the third line follows because the blocks of honest parties span the whole timespan without gaps by definition.

\textbf{Case $\bcweight(\abchain) \leq \xi^2 \cweight(\ares)$:} 
By definition of every block $b_i$ with timespan $(t_i, t'_i)$, it follows that
\begin{align*}
        \bcweight(\abchain) &= \sum_{b_i \in \abchain} \weight\left(\bSb^\Ac(b_i), \bVb^\Ac(b_i), \bWb^\Ac(b_i)\right) \\
                            &\leq \sum_{b_i \in \abchain} \weight\left(\bSb^\Ac(b_i), \int_{t_i}^{t'_i} \Vb^\Ac(t)\,dt, \int_{t_i}^{t'_i} \Wb^\Ac(t)\,dt\right)\\
                            &\leq \sum_{b_i \in \abchain} \weight\left(\bSbmax^\Ac(b_i), \int_{t_i}^{t'_i} \Vb^\Ac(t)\,dt, \int_{t_i}^{t'_i} \Wb^\Ac(t)\,dt\right).
\end{align*}
The third line follows by the monotonicity of $\weight$ and the fact that $\bSb^\Ac(b_i) \leq \bSbmax^\Ac(b_i)$ necessarily. 

Using the previous insights about the average resources, we continue with
\begin{align*}
        \bcweight(\abchain) &\leq \sum_{b_i \in \abchain} \weight\left(\bSbmax^\Ac(b_i), \int_{t_i}^{t'_i} \Vb^\Ac(t)\,dt, \int_{t_i}^{t'_i} \Wb^\Ac(t)\,dt\right)\\
                            &= \sum_{b_i \in \abchain} (t'_i-t_i) \cdot \weight\left(\bSbmax^\Ac(b_i), \overline{\bVb}^\Ac(b_i), \overline{\bWb}^\Ac(b_i)\right)\\
                            &\leq \sum_{b_i \in \abchain} (t'_i-t_i) \cdot \weight\left(\bSbmax^\Ac(b_i), \bVbmax^\Ac(b_i), \bWbmax^\Ac(b_i)\right)
\end{align*}
where the last line follows from monotonicity.

Now we invoke \cref{def:smooth} to switch $\max$ to $\min$, that is,
\begin{align*}
        \bcweight(\abchain) &\leq \sum_{b_i \in \abchain} (t'_i-t_i) \cdot \weight\left(\bSbmax^\Ac(b_i), \bVbmax^\Ac(b_i), \bWbmax^\Ac(b_i)\right)\\
                            &\leq \xi \sum_{b_i \in \abchain} (t'_i-t_i) \cdot \weight\left(\bSbmin^\Ac(b_i), \bVbmax^\Ac(b_i), \bWbmax^\Ac(b_i)\right)\\
                            &= \xi^2 \sum_{b_i \in \abchain} (t'_i-t_i) \cdot \weight\left(\bSbmin^\Ac(b_i), \bVbmin^\Ac(b_i), \bWbmin^\Ac(b_i)\right)
\end{align*}
where the second line follows from the sub-homogeneity of $\weight$ in $\Sb$ and the third from the homogeneity of $\weight$ in $(\Vb, \Wb)$.

This implies the desired inequality because
\begin{align*}
        \bcweight(\abchain) &\leq \xi^2 \sum_{b_i \in \abchain} (t'_i-t_i) \cdot \weight(\bSbmin^\Ac(b_i), \bVbmin^\Ac(b_i), \bWbmin^\Ac(b_i))\\
                            &\leq \xi^2 \sum_{b_i \in \abchain} \int_{t_i}^{t'_i}\weight(\Sb^\Ac(t), \Vb^\Ac(t), \Wb^\Ac(t)) \, dt\\
                            &\leq \xi^2 \int_0^T \weight(\Sb^\Ac(t), \Vb^\Ac(t), \Wb^\Ac(t)) \, dt\\
                            &= \xi^2 \cweight(\ares).
\end{align*}
Note that the third line follows because the adversary may leave some gaps in time between blocks.
\end{proof}

\end{document}